\documentclass[11pt]{imsart}
\usepackage[margin=1in]{geometry}
\usepackage[T1]{fontenc}
\usepackage{amsmath,amsthm,amssymb,mathtools}
\usepackage{graphicx}
\usepackage{xcolor}
\usepackage{algorithm}
\usepackage{algpseudocode}
\usepackage{array,longtable}
\usepackage{placeins}
\usepackage[authoryear,round]{natbib}
\usepackage[hidelinks]{hyperref}
\startlocaldefs
\numberwithin{equation}{section}
\theoremstyle{plain}
\newtheorem{theorem}{Theorem}[section]
\newtheorem{proposition}[theorem]{Proposition}
\newtheorem{lemma}[theorem]{Lemma}
\newtheorem{corollary}[theorem]{Corollary}
\theoremstyle{definition}

\newtheorem{example}[theorem]{Example}
\theoremstyle{remark}
\newtheorem{remark}[theorem]{Remark}
\DeclareMathOperator{\Tr}{Tr}
\DeclareMathOperator{\rank}{rank}
\DeclareMathOperator{\diag}{diag}
\newcommand{\M}{M_n(\mathbb C)}
\newcommand{\id}{\mathrm{id}}
\newcommand{\cut}{\mathrm{cut}}
\newcommand{\proj}{\mathrm{proj}}

\newcommand{\ip}[2]{\left\langle #1,#2\right\rangle}
\DeclareRobustCommand{\limitation}[1]{#1}
\DeclareRobustCommand{\supporting}[1]{#1}
\endlocaldefs
\allowdisplaybreaks

\begin{document}
\begin{frontmatter}
\title{Projection discrepancy and collective tests for unital quantum channels}
\runtitle{Projection discrepancy of quantum channels}
\begin{aug}
\author[A]{\fnms{Ning}~\snm{Ning}\thanksref{corresponding}\ead[label=e1]{patning@tamu.edu}}
\thankstext{corresponding}{Corresponding author.}
\runauthor{N. Ning}
\address[A]{Department of Statistics, Texas A\&M University,
College Station, Texas, USA, \printead{e1}}
\end{aug}

\begin{abstract}
Determining how effectively a quantum process randomizes its inputs is
important for assessing mixing and designing measurements that reveal
residual information. Spectral measures summarize retained information
globally, but do not directly identify which preparations and measurements
can detect it. This motivates a quantitative connection between spectral
behavior and explicit preparation-and-measurement tests.
We study quantum channels that preserve the maximally mixed state and
use a correlation score normalized by the sizes of the tested subspaces.
We establish sharp bounds relating this score to spectral retention,
construct channels attaining the bounds, and identify the parameter
ranges in which equality is possible. A reduction to classical
correlations connects established mixing inequalities with quantum
projection tests. From a supplied channel model and a suitable input,
an explicit finite procedure selects tests that certify the resulting
bounds. We further show how independent parallel uses can increase the
score through separable preparations and measurements performed
separately on each use.
Numerical studies support both the sharp bounds and the collective-testing
construction. Exhaustive calculations in small dimensions and independent
qubit calculations agree with the analytic predictions. For two synthetic
channel families, tests using 64 independent channel uses produce estimated
score gains of approximately $20.3\%$ and $5.9\%$ over the respective exact
one-use optima, with simultaneous $99\%$ confidence intervals entirely
above those baselines.
\end{abstract}
\begin{keyword}
\kwd{quantum channels}
\kwd{quantum expanders}
\kwd{projection discrepancy}
\kwd{quasirandomness}
\end{keyword}
\end{frontmatter}
\pagestyle{plain}

\section{Introduction}
\label{sec:introduction}

We introduce the background and motivation for studying projection tests
of unital quantum channels, summarize our contributions, and outline
the structure of the paper.

\subsection{Background and motivation}
\label{sec:intro-background}

Quantum state randomization has a direct application in cryptography:
\citet{ambainis2004small} constructed explicit approximate encryption
schemes that conceal quantum messages using short classical keys.
For quantum mixing, \citet{hastings2007random} showed that random
unitary channels can have a spectral gap bounded away from zero
while using a number of unitary conjugations independent of the
state-space dimension.
These constructions motivate quantitative tests of the information
that survives a channel. For unital channels, which preserve the
maximally mixed state, \citet{bannink2020quasirandom} related spectral
expansion to correlations between subspaces. They proved an equivalence
under irreducible covariance, but exhibited sequences of channels
outside this class whose maximum subspace correlation, without
normalization by subspace size, tends to zero while spectral retention
stays bounded away from zero. This separation
motivates normalizing correlations by the subspace sizes and seeking
explicit tests that witness the resulting bounds. We therefore ask:
\emph{How large must the maximum rank-normalized correlation between input
and output subspaces be at a given spectral parameter? Which projection
tests witness these lower bounds? How do independent parallel uses change the
optimized discrepancy?}

Fix an integer $n\geq2$. Let $M_n(\mathbb C)$ be the space of complex
$n\times n$ matrices. Let $I_n$ denote the $n\times n$ identity matrix,
abbreviated to $I$ when the dimension is clear. For $X\in M_n(\mathbb C)$,
write $X^\dagger$ for its conjugate transpose and $\Tr(X)$ for its ordinary
trace. Define the scalar normalized trace $\tau(X)=\Tr(X)/n$ and the
completely depolarizing map $\Pi(X)=\tau(X)I_n\in M_n(\mathbb C)$.
The normalized Hilbert--Schmidt norm is
$\|X\|_2=(\tau(X^\dagger X))^{1/2}$; for a complex-linear map
$S:M_n(\mathbb C)\to M_n(\mathbb C)$, its induced norm is
$\|S\|_{2\to2}=\sup_{X\ne0}\|S(X)\|_2/\|X\|_2$.
Projection symbols $P,Q$ denote orthogonal projections on $\mathbb C^n$,
and $\rank(P)$ is the dimension of the range of $P$.
For a unital quantum channel $\Phi:M_n(\mathbb C)\to M_n(\mathbb C)$, define
\begin{equation}
\lambda(\Phi)=\|\Phi-\Pi\|_{2\to2} \quad\text{and}\quad \delta(\Phi)=
 \max_{P,Q\ne0\text{ projections}}
 \frac{|\tau[Q(\Phi-\Pi)(P)]|}{\sqrt{\tau(P)\tau(Q)}}.
 \label{eq:intro-lambda}
\end{equation}
We omit the channel argument in these and subsequent discrepancy
parameters when $\Phi$ is fixed. The input and output projections
need not coincide. The factors of $n$ cancel in the second
formula. For nonzero projections $P,Q$ with ranks $r=\rank(P)$ and
$s=\rank(Q)$, the Born
bias $|\Tr[Q\Phi(P/r)]-s/n|$ for input $P/r$ is its discrepancy
score times $\sqrt{s/r}$.
We study this rank-weighted correlation functional. Its
amplification does not by itself establish better discrimination
at a fixed channel-use budget or a procedure for learning an
unknown channel.
More precisely, the cut parameter of \citet{bannink2020quasirandom} is
\begin{equation}
 \epsilon(\Phi)=\|\Phi-\Pi\|_{\cut}
 :=\max_{P,Q\text{ projections}}|\tau[Q(\Phi-\Pi)(P)]|.
 \label{eq:intro-epsilon}
\end{equation}
Their linear converse under irreducible covariance is
$\lambda\leq2\pi^2\epsilon$
\citep[Corollary~3.4]{bannink2020quasirandom}. We study the stronger
rank-normalized parameter $\delta$ without a covariance
assumption.

The underlying comparison mechanisms are classical. The restriction of
$\delta$ is Peyre's $\tau$-mixing coefficient; his event-envelope theorem
yields the sharp dimension constants used here, and his
variance-normalized comparison yields \eqref{eq:intro-main}
\citep{Peyre2013}. The logarithmic comparison and the necessity of its
order have earlier antecedents in
\citet{BradleyBryc1985,BradleyBrycJanson1987}; see
\citet[p.~2]{Peyre2013} also for the independent work
of Bulinskii. Peyre identifies the sharp coefficient used here.
Related matrix and weighted comparisons appear in
\citet{BollobasNikiforov2004,Butler2006Discrepancy,Nikiforov2009CutNorms,lev2015discrete}.
The direct spectral proofs here supply finite witnesses, equality-channel
realizations, and the precise endpoints of the linear branches.
The parallel construction implements Witsenhausen's rounding principle
by separable preparations and measurements local across uses
\citep{Witsenhausen1975,GhaziKamathSudan2016}, with a direct tracial
proof of maximal-correlation tensorization \citep{Beigi2013}.
Related height estimates belong to the converse-mixing tradition of
\citet{bilu2006lifts}.

\subsection{Contributions}
\label{sec:intro-contributions}

Throughout, $\log$ denotes the natural logarithm, $\lfloor x\rfloor$
and $\lceil x\rceil$ are the floor and ceiling of a real number $x$,
and $O(1)$ denotes a quantity bounded as the dimension tends to infinity.
Our one-use comparison has an exact dimension constant $K_n$:
\begin{equation}
 \frac{\lambda}{K_n}\leq\delta\leq(1-1/n)\lambda\quad\text{and}\quad K_n=2\sum_{j=1}^{\lfloor n/2\rfloor}
                     (\sqrt j-\sqrt{j-1})^2
          =\tfrac12\log n+O(1).
 \label{eq:intro-dimension}
\end{equation}
The lower constant is optimal in every dimension. Signed
square-root-increment spectra give channels with
$\lambda=K_n/n$, $\delta=1/n$, and mixing with $\Pi$ yields
the exact lower envelope on $[0,K_n/n]$. Equality rigidity for
classical compressions proves that this linear branch ends exactly at
$K_n/n$. The proof constructs
at most $(n-1)^2$ complemented spectral-prefix pairs once a
traceless Hermitian norm-attaining input is known.
Algorithm~\ref{alg:supplied-witness} gives the supplied-model
procedure below, with its definitions, guarantee, and access assumptions.
Theorem~\ref{thm:centered-dimension} proves the guarantee.

\paragraph{A finite projection-witness construction.}
Algorithm~\ref{alg:supplied-witness} converts a supplied channel model and
input direction into explicit preparation and measurement tests, together
with a certified lower bound on the projection discrepancy.
For a unital channel, put $S=\Phi-\Pi$. More generally, Algorithm~\ref{alg:supplied-witness}
applies to a complex-linear map $S:M_n(\mathbb C)\to M_n(\mathbb C)$
that takes Hermitian matrices to Hermitian matrices and is centered:
$S(I)=0$ and $\Tr(S(X))=0$ for every $X$.
The supplied input is a nonzero traceless Hermitian matrix
$H\in M_n(\mathbb C)$;
write $Y=S(H)$, which is also traceless and Hermitian.
For $Y\ne0$, let $(u_i)_{i=1}^n$ and $(v_j)_{j=1}^n$ be orthonormal
eigenbases of $H$ and $Y$, respectively, in decreasing eigenvalue order.
For integers $1\leq k,l\leq n-1$, define their spectral-prefix projections by
\[
 P_k=\sum_{i=1}^k u_i u_i^\dagger\quad\text{and}\quad Q_l=\sum_{j=1}^l v_j v_j^\dagger.
\]
Their complemented prefixes and ranks are
\[
 R_k=\begin{cases}P_k,&k\leq n/2,\\ I-P_k,&k>n/2,\end{cases}
 \qquad
 T_l=\begin{cases}Q_l,&l\leq n/2,\\ I-Q_l,&l>n/2,\end{cases}
 \qquad
 r_k=\min(k,n-k),\quad s_l=\min(l,n-l).
\]
Thus $\rank(R_k)=r_k$ and $\rank(T_l)=s_l$.
Any orthonormal eigenbasis within a repeated-eigenvalue subspace may be used.

\begin{algorithm}[!t]
\caption{Supplied-model projection witness}
\label{alg:supplied-witness}
\begin{algorithmic}[1]
\Require A centered Hermiticity-preserving map $S$ and a supplied
nonzero traceless Hermitian input $H$, as defined in the accompanying text.
\Ensure A score $d_*$ and, when $Y\ne0$, projections $(P,Q)$ attaining it;
for $S=\Phi-\Pi$, the score certifies $\delta(\Phi)\geq d_*$.
\State $Y\gets S(H)$.
\If{$Y=0$}
 \State \Return $d_*=0$, a zero lower-bound certificate from this input.
\EndIf
\State Diagonalize $H$ and $Y$ to obtain the ordered eigenbases $(u_i)$ and $(v_j)$.
\State Form $P_k,Q_l,R_k,T_l$ and ranks $r_k,s_l$ using the accompanying definitions.
\State Compute the normalized correlation score
$d_{kl}\gets |\Tr[T_lS(R_k)]|/\sqrt{r_ks_l}$ for every $1\leq k,l\leq n-1$.
\State Choose indices $(k_*,l_*)$ maximizing $d_{kl}$, breaking ties arbitrarily.
\State \Return $P=R_{k_*}$, $Q=T_{l_*}$, and $d_*=d_{k_*l_*}$.
\end{algorithmic}
\end{algorithm}

\paragraph{One-use guarantees and classical compression.}
Algorithm~\ref{alg:supplied-witness} maximizes the rank-normalized
correlation among at most $(n-1)^2$ candidate pairs. For $S=\Phi-\Pi$,
the returned projections provide an explicit preparation-and-measurement
test whose score $d_*$ certifies $\delta(\Phi)\geq d_*$.
By Theorem~\ref{thm:centered-dimension}, for $Y\ne0$ its returned score is
at least $\|S(H)\|_2/(K_n\|H\|_2)$, with $K_n$ defined in
\eqref{eq:intro-dimension}. In particular, for an input-quality parameter
$\gamma\in(0,1]$, if the supplied input satisfies
$\|S(H)\|_2\geq\gamma\|S\|_{2\to2}\|H\|_2$,
the score is at least $\gamma\|S\|_{2\to2}/K_n$.
For a unital channel and a norm-maximizing input, this certifies
the lower bound $\lambda/K_n$.
Proposition~\ref{prop:model-stability} gives a certificate when the supplied
channel model has bounded error.
A two-basis compression explains the classical mechanism behind these comparisons.
Diagonalizing a Hermitian maximizing input and its output
produces a doubly stochastic matrix with the same $\lambda$.
Its event tests lift to channel projections, so compression does
not increase $\delta$ or $\epsilon$. Their minimum values
at fixed $\lambda$ are therefore attained by classical channels.
For the function $f:[0,1]\to[0,1]$ defined below,
Peyre's event-correlation inequality gives
\begin{equation}
 \lambda\leq f(\min\{2\delta,1\}),\qquad
 f(t)=t(1-\log t)\ (0<t\leq1),\quad f(0)=0.
 \label{eq:intro-main}
\end{equation}
Consequently $\delta\to0$ if and only if $\lambda\to0$,
even along changing dimensions. The exact increment family
shows that the logarithmic order cannot be removed. 

\paragraph{Collective tests under parallel channel uses.}
Parallel uses permit a different comparison. For a positive integer $k$,
let $\Phi^{\otimes k}$ denote the tensor-product channel representing
$k$ independent uses of $\Phi$ on $M_{n^k}(\mathbb C)$. Write
$\delta_k=\delta(\Phi^{\otimes k})$ and
$\epsilon_k=\epsilon(\Phi^{\otimes k})$, with limits
$\delta_\infty=\delta_\infty(\Phi)=\lim_{k\to\infty}\delta_k$ and
$\epsilon_\infty=\epsilon_\infty(\Phi)=\lim_{k\to\infty}\epsilon_k$.
These sequences are nondecreasing under identity padding, and
their limits satisfy
\begin{equation}
 \delta_\infty\geq\frac{\arcsin\lambda}{\pi}\quad\text{and}\quad
 \epsilon_\infty\geq\frac{\arcsin\lambda}{2\pi}.
 \label{eq:intro-regularized}
\end{equation}
The construction uses a separable, generally classically
correlated input ensemble and local measurements across channel
uses, followed by classical thresholding. It can exceed the
one-use optimum by an unbounded factor as dimension grows.
A chosen threshold sequence need not improve monotonically,
and these lower bounds need not equal the optimized limits.
Among $k$-fold tensor powers of unital channels on $M_n(\mathbb C)$ with
$\lambda(\Phi)=1$, the minimum of $\delta(\Phi^{\otimes k})$
is $1-(\lfloor n/2\rfloor/n)^k$, attained for every $k$ by
the same balanced block channel. For the binary symmetric
channel of correlation $\rho\in[0,1]$, already
$\delta_2=\rho/2+\rho^2/4>\delta_1$ when $\rho>0$, whereas
$\epsilon_k=\rho/4$ for all $k$.

\paragraph{Numerical evidence.}
The experiments test both the sharp one-use comparisons and the benefit
of collective threshold tests. \supporting{Exhaustive subset calculations
in dimensions $2$--$11$ agree with the spectral-profile formulas within
$6\times10^{-17}$ over 120 cases. Independent Kraus calculations for six
qubit channels agree with the optimal measurement-axis predictions within
$2\times10^{-16}$.} These checks support the implementation of the sharp
bounds and the role of independently chosen input and output projections.
For parallel uses, we measure improvement relative to the exact optimized
one-use discrepancy.
\supporting{At 64 uses, separable preparations and local measurements yield
estimated discrepancy gains of $20.3\%$ for the increment channel and
$5.9\%$ for the three-level staircase channel. Both simultaneous $99\%$
confidence intervals lie entirely above their one-use optima
(Table~\ref{tab:intro-numerics}). The same statistically supported
improvement occurs at 16 and 128 uses in both models.}
These findings demonstrate finite-copy gains in the rank-normalized
objective. Section~\ref{sec:experiments} reports all five copy counts.

\begin{table}[!t]
\centering\small
\caption{Numerical evidence for collective threshold tests at 64 uses.
Scores are Monte Carlo estimates for explicit balanced threshold effects.
Gains are relative to the exact one-use optimum $\delta_1$.
The intervals concern the specified test scores and have simultaneous
$99\%$ coverage across all ten settings in the full study; endpoints are
rounded outward. See Section~\ref{sec:balanced-effects} for the sampling
design and all copy counts.}
\label{tab:intro-numerics}
\begin{tabular}{lrrrr}
\hline
Model & $\delta_1$ & 64-use score & Gain & $99\%$ interval\\
\hline
Increment, $n=64$ & $0.015625$ & \supporting{$0.01879300$}
 & \supporting{$20.3\%$} & \supporting{$[0.017414,0.020172]$}\\
Staircase, $n=32$ & $0.031250$ & \supporting{$0.03309313$}
 & \supporting{$5.9\%$} & \supporting{$[0.032403,0.033783]$}\\
\hline
\end{tabular}
\end{table}
\FloatBarrier

\subsection{Structure of the paper}
\label{sec:intro-structure}

Section~\ref{sec:preliminaries} fixes the quantum-information
conventions. Section~\ref{sec:mixing} proves classical compression
and the one-use converse together. Section~\ref{sec:dimension}
gives the exact dimension bounds and witnesses, and
Section~\ref{sec:parallel} treats parallel tests.
Section~\ref{sec:experiments} compares the parameters and
illustrates the constructions.
Appendix~\ref{sec:unrestricted} proves the sharp unrestricted-map
comparison; Appendices~\ref{sec:proof} and~\ref{sec:general} give
height bounds by diagonal compression and projection witnesses for
arbitrary superoperators.
Table~\ref{tab:notation} at the end collects the notation.

\section{Quantum channels and discrepancy parameters}
\label{sec:preliminaries}

\subsection{States, measurements, and norms}
\label{sec:norms}

Let $\M$ be the operators on $\mathbb C^n$, with $n\geq2$, and let
$I$ denote its identity. A density operator is a positive semidefinite
matrix $\rho$ with $\Tr(\rho)=1$. An effect is an operator $0\leq F\leq I$;
the two-outcome measurement $\{F,I-F\}$ accepts a state $\rho$ with
probability $\Tr(F\rho)$. For a nonzero projection $P$, the state
$\rho_P=P/\rank(P)$ is maximally mixed on its range.
Fix an orthonormal basis $(|i\rangle)_{i=1}^n$ of $\mathbb C^n$,
with $\langle i|=|i\rangle^\dagger$, and write
$X_{ij}=\langle i|X|j\rangle$ for matrix entries in this basis.
The symbols $X^\top$ and $\overline X$ denote transpose and entrywise
conjugation, respectively. For a vector $z\in\mathbb C^n$,
$\diag(z)$ is the diagonal matrix with entries $z_1,\ldots,z_n$;
for a matrix $X$, $\diag(X)$ is its vector of diagonal entries.

Following the normalized conventions of
\citet{bannink2020quasirandom}, for $X,Y\in\M$ we use the inner product
\[
 \ip{X}{Y}=\tau(X^\dagger Y)
\]
and, with the matrix absolute value $|X|=(X^\dagger X)^{1/2}$,
the Schatten norms
\[
 \|X\|_p=\bigl(\tau(|X|^p)\bigr)^{1/p}\quad(1\leq p<\infty)\quad\text{and}\quad
 \|X\|_\infty=\text{the largest singular value of }X.
\]
In particular, $\|I\|_p=1$ and $\|P\|_2=\sqrt{\tau(P)}$.
We also use the ordinary Frobenius norm
$\|X\|_{\mathrm F}=(\Tr(X^\dagger X))^{1/2}=\sqrt n\,\|X\|_2$;
the factor $\sqrt n$ cancels in input/output norm ratios.
The ordinary trace is used for density operators
and Born probabilities. For a complex-linear map $S:\M\to\M$
and exponents $p,q\in[1,\infty]$, let
\[
 \|S\|_{p\to q}=\sup_{X\ne0}\frac{\|S(X)\|_q}{\|X\|_p}.
\]
For $p=q$, these induced norms agree with those defined using
unnormalized Schatten norms. The adjoint $S^*$ is characterized by
$\ip{Y}{S(X)}=\ip{S^*(Y)}{X}$.
For a vector $x=(x_1,\ldots,x_n)\in\mathbb C^n$,
$\|x\|_{\ell_p}=(\sum_i|x_i|^p)^{1/p}$ for
$1\leq p<\infty$, and $\|x\|_{\ell_\infty}=\max_i|x_i|$.
For a matrix $B\in M_n(\mathbb C)$, the corresponding induced norm is
$\|B\|_{\ell_p\to\ell_q}=\sup_{x\ne0}\|Bx\|_{\ell_q}/\|x\|_{\ell_p}$;
$\|B\|_{\mathrm{op}}$ means
$\|B\|_{\ell_2\to\ell_2}$.
All logarithms are natural.

A map $S:\M\to\M$ is Hermiticity preserving if
$S(X^\dagger)=S(X)^\dagger$ for all $X$. Differences of quantum
channels have this property.

\begin{lemma}[Hermitian norm attainment]
\label{lem:hermitian-attainment}
The $2\to2$ norm of a Hermiticity-preserving complex-linear map
$S:\M\to\M$ is attained at a Hermitian input.
Moreover, one of the nonzero Hermitian parts of any nonzero matrix
$X$ attains an input/output norm ratio at least
$\|S(X)\|_2/\|X\|_2$.
\end{lemma}

\begin{proof}
For $X\ne0$, write $X=A+iB$, where
$A=(X+X^\dagger)/2$ and $B=(X-X^\dagger)/(2i)$ are Hermitian.
Cyclicity of the trace gives $\Tr(AB-BA)=0$, hence
$\|X\|_{\mathrm F}^2=\|A\|_{\mathrm F}^2+\|B\|_{\mathrm F}^2$.
Hermiticity preservation gives the same identity after applying $S$.
Consequently
\[
 \frac{\|S(X)\|_{\mathrm F}^2}{\|X\|_{\mathrm F}^2}
 =\frac{\|S(A)\|_{\mathrm F}^2+\|S(B)\|_{\mathrm F}^2}
        {\|A\|_{\mathrm F}^2+\|B\|_{\mathrm F}^2}
 \leq\max_{\substack{C\in\{A,B\}\\C\ne0}}
             \frac{\|S(C)\|_{\mathrm F}^2}{\|C\|_{\mathrm F}^2}.
\]
The normalization factor $\sqrt n$ cancels in each ratio, so the
same assertion holds for $\|\cdot\|_2$. Thus the full supremum
equals the supremum over Hermitian inputs. The map
$H\mapsto\|S(H)\|_2$ is continuous on the compact Hermitian unit
sphere, so this supremum is a maximum.
\end{proof}

\subsection{Unital channels and the centered map}

Complex-linear maps on $\M$ are called superoperators.
A quantum channel is a completely positive, trace-preserving
superoperator $\Phi:\M\to\M$. It admits a representation
using a finite family of Kraus matrices $K_a\in\M$, indexed by $a$:
\begin{equation}
 \Phi(X)=\sum_a K_aXK_a^\dagger\quad\text{and}\quad
 \sum_a K_a^\dagger K_a=I.
 \label{eq:kraus}
\end{equation}
We assume in addition that $\Phi$ is \emph{unital}, or equivalently
$\sum_a K_aK_a^\dagger=I$. Such a channel is also called bistochastic.
It fixes the maximally mixed state $I/n$. The completely depolarizing
channel is
\[
 \Pi(X)=\tau(X)I.
\]
It is the orthogonal projection onto $\mathbb CI$ in the normalized
Hilbert--Schmidt inner product.
For any ambient dimension $N\geq1$, let $I_N$ be the identity in
$M_N(\mathbb C)$. For $X\in M_N(\mathbb C)$, we write
$\tau_N(X)=\Tr(X)/N$ and $\Pi_N(X)=\tau_N(X)I_N$.
Unsubscripted $\tau$, $\Pi$, and $I$ refer to the ambient matrix algebra;
for $k$ parallel uses its dimension is $n^k$.

Given a positive integer $d$, unitary matrices $U_1,\ldots,U_d\in\M$,
and probabilities $p_1,\ldots,p_d$, a mixed-unitary channel has the form
\begin{equation}
 \Phi(X)=\sum_{j=1}^d p_jU_jXU_j^\dagger,\qquad
 p_j\geq0,\quad \sum_jp_j=1,
 \label{eq:mixed-unitary}
\end{equation}
and is unital. The equal-weight case $p_j=1/d$ is the usual
random-unitary model for quantum expanders
\citep{hastings2007random,ben2008quantum,pisier2014quantum}.
A bounded-degree quantum expander family in this model consists of
channels $\Phi_k$ indexed by positive integers $k$, on dimensions tending
to infinity with fixed $d$ and $\sup_k\lambda(\Phi_k)<1$.
The estimates below hold for each channel
individually.

\begin{lemma}
\label{lem:channel-norms}
Let $I^\perp=\{X\in\M:\Tr(X)=0\}$ be the traceless subspace,
and write $\Phi|_{I^\perp}$ for the restriction of $\Phi$ to it.
If $\Phi$ is a unital quantum channel and $S=\Phi-\Pi$, then
\begin{align}
 &\Phi\Pi=\Pi\Phi=\Pi,\qquad S(I)=0,\qquad S^*(I)=0,
 \label{eq:center-identities}\\
 &\|\Phi\|_{1\to1}=\|\Phi\|_{2\to2}
      =\|\Phi\|_{\infty\to\infty}=1,
 \label{eq:channel-contraction}\\
 &\|S\|_{1\to1}\leq2,\qquad
   \|S\|_{\infty\to\infty}\leq2,\qquad
 \lambda(\Phi)=\|S\|_{2\to2}
      =\|\Phi|_{I^\perp}\|_{2\to2}\leq1,
 \label{eq:center-norms}
\end{align}
\end{lemma}

\begin{proof}
Unitality gives $\Phi\Pi=\Pi$, and trace preservation gives
$\Pi\Phi=\Pi$. A unital completely positive map satisfies the
Schwarz inequality
\[
 \Phi(X)^\dagger\Phi(X)\leq\Phi(X^\dagger X).
\]
For completeness, this follows by applying $\Phi$ entrywise to the
positive block matrix
$\left(\begin{smallmatrix}X^\dagger X&X^\dagger\\X&I\end{smallmatrix}\right)$
and taking the Schur complement of its bottom-right block $I$.
Since $X^\dagger X\leq\|X\|_\infty^2I$, it follows that
$\|\Phi(X)\|_\infty\leq\|X\|_\infty$ for every complex matrix $X$.
The adjoint $\Phi^*$ is also a unital quantum channel. Schatten norm
duality consequently yields
$\|\Phi\|_{1\to1}=\|\Phi^*\|_{\infty\to\infty}=1$.
Taking the trace in the Schwarz inequality proves
$\|\Phi(X)\|_2\leq\|X\|_2$. All three norms attain one at $I$.

These facts apply to $\Pi$ as well, so the triangle inequality gives the
two endpoint bounds for $S$. The orthogonal decomposition
\[
 \M=\mathbb CI\oplus I^\perp
\]
is invariant under $\Phi$; relative to it, $S$ is zero on the first
summand and equals $\Phi$ on the second. This proves the remaining
norm identity and bound.
\end{proof}

\begin{remark}
The parameter $\lambda(\Phi)$ is a singular-value parameter. It need
not equal the spectral radius of $\Phi|_{I^\perp}$ unless that
restriction is normal. For a Hilbert--Schmidt self-adjoint channel it
equals the largest modulus of a nontrivial eigenvalue. We do not
assume self-adjointness, primitivity, or that $I$ is the only fixed
operator.
\end{remark}

\subsection{Projection discrepancy and cut discrepancy}

For a superoperator $S$, define
\begin{equation}
 \|S\|_{\proj}
   =\max_{P,Q\ne0\text{ projections}}
      \frac{|\ip{Q}{S(P)}|}{\|P\|_2\|Q\|_2}.
 \label{eq:projection-norm}
\end{equation}
The maxima exist: projections of each fixed rank form a compact set,
and there are only finitely many ranks. The channel parameter is
$\delta(\Phi)=\|\Phi-\Pi\|_{\proj}$, while $\epsilon(\Phi)$ is
defined by \eqref{eq:intro-epsilon}. Explicitly, for $r=\rank(P)$ and
$s=\rank(Q)$,
\begin{equation}
 \frac{|\ip{Q}{(\Phi-\Pi)(P)}|}{\|P\|_2\|Q\|_2}
 =
 \frac{|\Tr[Q\Phi(P)]-rs/n|}{\sqrt{rs}}.
 \label{eq:rank-formula}
\end{equation}
Thus $\delta$ tests deviations at every pair of subspace dimensions.
The cut discrepancy instead divides the ordinary trace deviation by
the ambient dimension $n$.

For comparison with centered event correlations, for each projection pair
put $p=\tau(P)$ and $q=\tau(Q)$ and define
\begin{equation}
 \zeta(\Phi)=
 \max_{\substack{P,Q\text{ projections}\\
                  0<\rank(P),\rank(Q)<n}}
 \frac{|\tau[Q(\Phi-\Pi)(P)]|}
 {\sqrt{p(1-p)q(1-q)}},
 \qquad p=\tau(P),\quad q=\tau(Q).
 \label{eq:zeta}
\end{equation}
Only nonzero proper projections occur, so the denominator is
positive. Compactness at each allowed rank again gives an
attained maximum. This is the channel version of Peyre's
variance-normalized event coefficient $\tau'$; the trace-normalized
$\delta$ instead corresponds to his $\tau$ coefficient
\citep[definitions (3)--(5)]{Peyre2013}.

There is also a bipartite-state interpretation. In the fixed basis,
let $E_{ij}=|i\rangle\langle j|$ for $1\leq i,j\leq n$ be the matrix units.
Write $J(\Phi)$ for the unnormalized Choi
matrix and $\omega_\Phi$ for its normalized state:
\begin{equation}
 J(\Phi)=\sum_{i,j}E_{ij}\otimes\Phi(E_{ij})\quad\text{and}\quad
 \omega_\Phi=\frac{J(\Phi)}n.
 \label{eq:choi-state}
\end{equation}
The standard Choi criterion identifies complete positivity with
$J(\Phi)\geq0$ \citep[Theorem~2.22]{Watrous2018}.
Trace preservation and unitality give both marginals $I/n$.
Expanding in matrix units gives, for $X,Y\in\M$ and
projections $P,Q$ with $p=\tau(P)$, $q=\tau(Q)$,
\begin{equation}
\Tr[\omega_\Phi(X\otimes Y^\dagger)]
    =\tau[Y^\dagger\Phi(X^\top)] \quad\text{and}\quad \Tr[\omega_\Phi(P^\top\otimes Q)]-pq
    =\tau[Q(\Phi-\Pi)(P)].
 \label{eq:choi-covariance}
\end{equation}
Thus $\delta,\epsilon,\zeta$ normalize the same connected
Choi-state projection correlations in three ways. Beigi's
quantum maximal correlation~\citep{Beigi2013} reduces here to
\begin{equation}
 \mu(\omega_\Phi)=
 \max_{\substack{\tau(X)=\tau(Y)=0\\\|X\|_2=\|Y\|_2=1}}
 |\tau[Y^\dagger\Phi(X^\top)]|=\lambda(\Phi).
 \label{eq:choi-maximal}
\end{equation}
Indeed, transpose preserves trace and the Hilbert--Schmidt norm,
and optimizing $Y$ gives the norm of $\Phi(X^\top)$ on
$I^\perp$. In this tracial case Hermitian maximizers are available
by Lemma~\ref{lem:hermitian-attainment}, consistently with
\citet[Remark~6]{Beigi2013}.
Proposition~\ref{prop:classical-compression} extracts a classical
law attaining the same correlation: its input eigenbasis is
measured after entrywise conjugation on the first Choi
subsystem, since the projection is $P^\top$ in
\eqref{eq:choi-covariance}.

Hypercontractivity gives complementary restrictions on local correlations:
\citet[Lemma~15]{DelgoshaBeigi2014} bound joint local-effect
probabilities through their marginals in a study of state transformations
from arbitrarily many copies. Their maximal-correlation comparison in
Theorem~16 assumes Hermitian maximizers, which are available in the
present tracial setting. These restrictions complement the explicit
threshold witnesses in Section~\ref{sec:parallel}.

\begin{proposition}[Effects and bistochastic data processing]
\label{prop:effects}
For every complex-linear map $S:\M\to\M$, with the supremum below
taken over nonzero effects $E,F\in\M$,
\begin{equation}
 \|S\|_{\proj}
 =\sup_{\substack{0\leq E,F\leq I\\E,F\ne0}}
   \frac{|\Tr[F S(E)]|}{\sqrt{\Tr(E)\Tr(F)}}.
 \label{eq:effects}
\end{equation}
If $\mathcal A,\mathcal B:\M\to\M$ are unital quantum channels, then
\begin{equation}
 \|\mathcal B S\mathcal A\|_{\proj}\leq\|S\|_{\proj}.
 \label{eq:projection-processing}
\end{equation}
Consequently, for every unital quantum channel $\Phi$,
\begin{equation}
 \delta(\mathcal B\Phi\mathcal A)\leq\delta(\Phi).
 \label{eq:channel-processing}
\end{equation}
\end{proposition}

\begin{proof}
For $t,u\in[0,1]$, let $P_t$ project onto the eigenspaces of $E$
with eigenvalues greater than $t$, and let $Q_u$ be the analogous
projection for $F$ at threshold $u$. The finite spectral decompositions
give $E=\int_0^1P_t\,dt$ and $F=\int_0^1Q_u\,du$.
By linearity and the triangle inequality,
\[
 |\Tr[F S(E)]|
 \leq\|S\|_{\proj}
 \left(\int_0^1\sqrt{\Tr(P_t)}\,dt\right)
 \left(\int_0^1\sqrt{\Tr(Q_u)}\,du\right).
\]
Zero threshold projections contribute zero. Cauchy--Schwarz on
$[0,1]$ bounds the two factors by $\sqrt{\Tr(E)}$ and
$\sqrt{\Tr(F)}$, respectively. These are integrals of finite step
functions, so no limiting interchange is needed. The reverse
inequality in \eqref{eq:effects} follows by restricting to projections.

For nonzero projections $P,Q$, the operators
$E=\mathcal A(P)$ and $F=\mathcal B^*(Q)$ are effects with
$\Tr(E)=\rank(P)$ and $\Tr(F)=\rank(Q)$. This uses both unitality
and trace preservation of $\mathcal A$ and $\mathcal B$.
Since
\[
 \Tr[Q\mathcal B S\mathcal A(P)]
   =\Tr[\mathcal B^*(Q)S(\mathcal A(P))],
\]
equation \eqref{eq:effects} proves \eqref{eq:projection-processing}.
Finally, $\mathcal B\Pi\mathcal A=\Pi$, and thus
$\mathcal B\Phi\mathcal A-\Pi=\mathcal B(\Phi-\Pi)\mathcal A$.
\end{proof}

The effect denominator in \eqref{eq:effects} is
$\sqrt{\Tr(E)\Tr(F)}$. With this normalization, projective tests
attain the same discrepancy as general effects.
Unitality is essential to the processing statement relative to
$\Pi$. Let $R$ be a rank-one projection and define the pure-state
replacer $\mathcal R(X)=\Tr(X)R$. Then $\mathcal R\Pi=\mathcal R$.
The test $P=Q=R$ gives
$\|\mathcal R\Pi-\Pi\|_{\proj}\geq1-1/n$, whereas
$\|\Pi-\Pi\|_{\proj}=0$. This nonunital postprocessing
also leaves the class of unital channels.

\section{Classical compression and the one-use converse}
\label{sec:mixing}

This section reduces the channel question to a finite classical
joint law without losing its spectral parameter. The reduction
both identifies the worst-case lower envelopes and transfers
Peyre's probability-dependent converse to projection tests.

\subsection{Classical embedding and spectral compression}

Let $A$ be a nonnegative doubly stochastic $n\times n$ matrix.
Its associated quantum channel is
\begin{equation}
 \Phi_A(X)=\sum_{i,j}A_{ij}\langle j|X|j\rangle
                      |i\rangle\langle i|.
 \label{eq:classical-channel}
\end{equation}
It measures in the computational basis, applies the stochastic
matrix, and prepares a computational basis state. Such
measure-and-prepare channels are entanglement breaking
\citep[Theorems~2 and~4]{HorodeckiShorRuskai2003}; the explicit
lower-envelope constructions and block-averaging examples below attain
their stated values within this subclass. The Kraus
operators $\sqrt{A_{ij}}\,|i\rangle\langle j|$ show that it is a
unital quantum channel. This embedding is also used in
\citet[Section~3.3]{bannink2020quasirandom}.
Write $[n]=\{1,\ldots,n\}$, let $\mathbf1$ be the all-ones vector,
and let $\mathbf1_E$ be the indicator vector of a subset $E\subseteq[n]$,
whose cardinality is denoted by $|E|$.
For the all-ones matrix $J_n=\mathbf1\mathbf1^\top$, we have
\begin{equation}
 \lambda(\Phi_A)=\|A-J_n/n\|_{\ell_2\to\ell_2}\quad\text{and}\quad
 \delta(\Phi_A)=
 \max_{\varnothing\ne S,T\subseteq[n]}
 \frac{|\mathbf1_T^\top(A-J_n/n)\mathbf1_S|}
      {\sqrt{|S||T|}}.
 \label{eq:classical-parameters}
\end{equation}
The first equality follows because the channel kills off-diagonal
inputs and acts as $A$ on the diagonal algebra. For the second, if
$P,Q$ have ranks $r,s$, their diagonals $p,q$ lie in $[0,1]^n$ and
sum to $r,s$. A real bilinear form attains its largest absolute
value on the product of these two polytopes at vertices. Since
$r,s$ are integers, these vertices are binary vectors: two fractional
coordinates would permit a small transfer between them while preserving
their sum, contradicting extremality, and an integer sum rules out
exactly one fractional coordinate. One can first maximize either sign
of the form over $p$ with $q$ fixed, choosing a vertex, and then over
$q$ with that vertex fixed.
Binary vectors are themselves diagonals of projections, proving
both inequalities needed for the equality.

The same cube-vertex argument, without rank denominators, gives
\begin{equation}
 \epsilon(\Phi_A)=
 \max_{S,T\subseteq[n]}
 \frac{|\mathbf1_T^\top(A-J_n/n)\mathbf1_S|}{n}.
 \label{eq:classical-cut}
\end{equation}
Let $J_{\mathrm{in}},J_{\mathrm{out}}$ be random labels in $[n]$ with
joint law $\Pr(J_{\mathrm{in}}=j,J_{\mathrm{out}}=i)=A_{ij}/n$,
where $\Pr$ denotes probability; we use $\mathbb E$ for expectation.
The labels have
uniform marginals. The second quantity in
\eqref{eq:classical-parameters} is precisely the event-based
$\tau$-mixing coefficient of \citet{Peyre2013}.

\begin{proposition}[A classical compression at fixed spectral parameter]
\label{prop:classical-compression}
For every unital channel $\Phi:\M\to\M$, there is a nonnegative
doubly stochastic matrix $A$ such that
\begin{equation}
 \lambda(\Phi_A)=\lambda(\Phi),\qquad
 \delta(\Phi_A)\leq\delta(\Phi),\qquad
 \epsilon(\Phi_A)\leq\epsilon(\Phi).
 \label{eq:compression-domination}
\end{equation}
Consequently, at every fixed $\ell\in[0,1]$, the minimum of
$\delta$ over unital channels with spectral parameter $\ell$
equals the corresponding minimum over classical channels
$\Phi_A$. The same statement holds for $\epsilon$.
All these minima are attained.
\end{proposition}

\begin{proof}
For $\lambda=0$, take $A=J_n/n$.
Otherwise, Lemma~\ref{lem:hermitian-attainment} gives a Hermitian
norm maximizer for $S=\Phi-\Pi$. Subtracting its scalar part
leaves its image unchanged because $S(I)=0$, and cannot increase
its norm. The resulting nonzero traceless matrix $H$ is therefore
also a maximizer. Let $U,V\in\M$ be unitary matrices whose columns
are orthonormal eigenbases of $H$ and $S(H)$, respectively.
Define
\begin{equation}
 A_{ij}=\Tr\!\left[
 V|i\rangle\langle i|V^\dagger\,
 \Phi(U|j\rangle\langle j|U^\dagger)\right].
 \label{eq:stochastic-compression}
\end{equation}
These are transition probabilities. Trace preservation and
unitality give column and row sums equal to one.
For any vector $z\in\mathbb C^n$, $(A-J_n/n)z$ is the diagonal of
$V^\dagger S(U\diag(z)U^\dagger)V$. Diagonal extraction
discards terms from the squared Hilbert--Schmidt norm. With
ordinary Euclidean vector norms this gives
\[
 \|(A-J_n/n)z\|_{\ell_2}
 \leq\sqrt n\,\|S(U\diag(z)U^\dagger)\|_2
 \leq\lambda(\Phi)\|z\|_{\ell_2}.
\]
For the eigenvalue vector of $H$, both the input and output
are diagonal in the chosen bases, so no norm is lost; its
ratio equals $\lambda(\Phi)$. This proves equality.
For subsets $E,F\subseteq[n]$, take
$P=U\diag(\mathbf1_E)U^\dagger$ and
$Q=V\diag(\mathbf1_F)V^\dagger$. Expanding the trace in the chosen
bases gives
\[
 \tau[Q S(P)]=\frac1n\mathbf1_F^\top(A-J_n/n)\mathbf1_E,
 \qquad \rank(P)=|E|,\quad \rank(Q)=|F|.
\]
For nonempty $E,F$, dividing by
$\sqrt{\tau(P)\tau(Q)}=\sqrt{|E||F|}/n$ gives exactly the
classical discrepancy score. Maximizing over subsets and using
\eqref{eq:classical-parameters} and \eqref{eq:classical-cut} proves
\eqref{eq:compression-domination}.

Embedded classical channels form a subset of the unital channels,
while the compression produces a classical channel with no larger
discrepancy at the same $\lambda$. These two directions prove
equality of the infima. The channel sets are compact in Choi
coordinates, the stochastic matrices form a compact polytope,
and the parameters are continuous maxima over compact projection
sets. Fixed-$\lambda$ level sets are therefore compact. They are
nonempty at every $\ell\in[0,1]$. Writing $\id(X)=X$ for the
identity map, examples are
$(1-\ell)\Pi+\ell\id$ and
$A=(1-\ell)J_n/n+\ell I$. The infima are thus attained.
\end{proof}

The bases depend on the channel and a norm-attaining input.
This tracial compression preserves the spectral parameter and may
decrease either discrepancy.
The lower-envelope equality uses both their nonincrease and the
inclusion of classical channels among unital channels.

\subsection{Mixing bounds and operational witnesses}

Define $f:[0,1]\to[0,1]$ by $f(0)=0$ and
$f(t)=t(1-\log t)$ for $t>0$. It is continuous and strictly
increasing, so its inverse $f^{-1}$ is well defined on $[0,1]$.

\begin{theorem}[Projection mixing and classical converse]
\label{thm:main}
Let $\Phi:\M\to\M$ be a unital quantum channel. For projections $P,Q$,
write $p=\tau(P)$ and $q=\tau(Q)$. Then
\begin{equation}
 |\tau[Q\Phi(P)]-pq|
 \leq\lambda(\Phi)\sqrt{p(1-p)q(1-q)}
 \leq\lambda(\Phi)\sqrt{pq}.
 \label{eq:forward}
\end{equation}
The discrepancy parameters satisfy
\begin{equation}
\delta(\Phi)\leq\zeta(\Phi)
       \leq\min\{2\delta(\Phi),\lambda(\Phi)\} \quad\text{and}\quad \lambda(\Phi)\leq f(\zeta(\Phi))
       \leq f(\min\{2\delta(\Phi),1\}).
 \label{eq:zeta-comparison}
\end{equation}
In particular, writing $\lambda=\lambda(\Phi)$,
\begin{equation}
 \frac{f^{-1}(\lambda)}2\leq\delta(\Phi)
       \leq\left(1-\frac1n\right)\lambda.
 \label{eq:channel-log-bounds}
\end{equation}
If $\lambda>0$, there are nonzero projections $P,Q$, each of
rank at most $n/2$, such that
\begin{equation}
 |\tau[Q\Phi(P)]-\tau(P)\tau(Q)|
 \geq \frac{f^{-1}(\lambda)}2\sqrt{\tau(P)\tau(Q)}.
 \label{eq:converse}
\end{equation}
If $\lambda=0$, then $\Phi=\Pi$ and
$\delta(\Phi)=\epsilon(\Phi)=\zeta(\Phi)=0$.
\end{theorem}

The forward estimate is the quantum expander mixing argument
of \citet{bannink2020quasirandom}. The converse below translates
\citet[Theorem~3.1]{Peyre2013} through the
preceding compression, with all probability normalizations explicit.

\begin{proof}
Let $S=\Phi-\Pi$, $P_0=P-pI$, and $Q_0=Q-qI$.
Using $S(I)=S^*(I)=0$, we obtain
\[
 \tau[Q\Phi(P)]-pq
   =\ip{Q}{S(P)}=\ip{Q_0}{S(P_0)}.
\]
Cauchy--Schwarz gives an upper bound
$\|S\|_{2\to2}\|P_0\|_2\|Q_0\|_2$.
Since $\|P_0\|_2^2=p(1-p)$ and
$\|Q_0\|_2^2=q(1-q)$, this proves \eqref{eq:forward}.
For a nonzero proper projection, $p\geq1/n$, and the same holds
for $q$. Dividing \eqref{eq:forward} by $\sqrt{pq}$ therefore gives
the upper bound in \eqref{eq:channel-log-bounds}. If either projection
is $I$, the centered pairing is zero.
In particular, a nonzero projection has $\Pi(P)=pI\ne0$; it is
$P_0$, not $P$, that is traceless.

For each proper projection pair, the $\zeta$ ratio is the
$\delta$ ratio divided by $\sqrt{(1-p)(1-q)}\leq1$, proving
$\delta\leq\zeta$; pairs involving $I$ have zero correlation.
Complementing either projection preserves the absolute centered
correlation and its Bernoulli variance. Thus a pair maximizing
$\zeta$ may be chosen with $p,q\leq1/2$, where its ratio is at
most twice the corresponding $\delta$ ratio. This proves
$\zeta\leq2\delta$. Finally, \eqref{eq:forward} gives
$\zeta\leq\lambda\leq1$.
For the converse, suppose $\lambda>0$ and use the matrix $A$ and bases $U,V$ from
Proposition~\ref{prop:classical-compression}. Under the finite
joint law $\Pr(J_{\mathrm{in}}=j,J_{\mathrm{out}}=i)=A_{ij}/n$, both labels are uniform.
For mean-zero real score vectors $x,y\in\mathbb R^n$ with
$n^{-1}\sum_jx_j^2=n^{-1}\sum_i y_i^2=1$, their covariance is
$n^{-1}y^\top Ax$. Its maximum absolute value is the classical
maximal correlation
\[
 \rho(A)=\|A-J_n/n\|_{\ell_2\to\ell_2}=\lambda(\Phi).
\]
Here the matrix $J_n$ is the all-ones matrix, as in
\eqref{eq:classical-parameters}; real singular vectors suffice
because $A-J_n/n$ is real. Let $\tau'(A)$ be the maximum
absolute covariance of nontrivial coordinate events divided by their
Bernoulli standard deviations. Every such pair lifts in the
bases $U,V$ to projections with the same covariance and
probabilities. Hence $\tau'(A)\leq\zeta(\Phi)$.

Peyre's classical inequality states
$\rho(A)\leq\tau'(A)(1-\log\tau'(A))$
\citep[Theorem~3.1]{Peyre2013}, with value
zero when $\tau'(A)=0$. It applies to two sigma algebras on
a common probability space, including the present finite law
with possibly zero atoms. Since $f$ is increasing on $[0,1]$,
\[
 \lambda(\Phi)\leq f(\tau'(A))\leq f(\zeta(\Phi))
 \leq f(\min\{2\delta(\Phi),1\}).
\]
This proves \eqref{eq:zeta-comparison}. Inverting $f$ and using
$\zeta\leq2\delta$ gives $\delta\geq f^{-1}(\lambda)/2$.
Together with the forward bound, this proves
\eqref{eq:channel-log-bounds}.
Choose a pair attaining $\delta$. A nonzero correlation
excludes $I$; complementing any projection of rank above
$n/2$ preserves the absolute covariance and decreases the
rank denominator. This gives \eqref{eq:converse} with the
stated rank restriction. When $\lambda=0$, $\Phi=\Pi$ and
all three discrepancies vanish.
\end{proof}

\begin{corollary}[Equivalent vanishing parameters]
\label{cor:vanishing}
For any sequence of unital channels $(\Phi_k)_{k\geq1}$, possibly on different
dimensions,
\[
 \delta(\Phi_k)\longrightarrow0
 \quad\Longleftrightarrow\quad
 \lambda(\Phi_k)\longrightarrow0.
\]
More quantitatively, for $\delta=\delta(\Phi)$,
\begin{equation}
 \lambda(\Phi)\leq f(\min\{2\delta,1\})
 =\begin{cases}
   2\delta(1-\log(2\delta)),&0<\delta\leq1/2,\\
   0,&\delta=0,\\
   1,&\delta\geq1/2.
  \end{cases}
 \label{eq:inverse-delta}
\end{equation}
\end{corollary}

\begin{proof}
If $\lambda(\Phi_k)\to0$, the bound
$0\leq\delta(\Phi_k)\leq\lambda(\Phi_k)$ gives
$\delta(\Phi_k)\to0$. Conversely, if $\delta(\Phi_k)\to0$, then
\eqref{eq:zeta-comparison} gives
\[
 0\leq\lambda(\Phi_k)
 \leq f(\min\{2\delta(\Phi_k),1\})\longrightarrow f(0)=0.
\]
Both comparisons are independent of the channel dimension.
Substituting the definition of $f$ at $\min\{2\delta,1\}$ gives
\eqref{eq:inverse-delta}, including the endpoint values $f(0)=0$
and $f(1)=1$.
\end{proof}

\begin{corollary}[Preparation and measurement witnesses]
\label{cor:operational}
If $\lambda(\Phi)>0$, there are nonzero projections $P,Q$, of ranks
$r,s$ respectively, for which the input state $\rho_P=P/r$ and
the two-outcome measurement $\{Q,I-Q\}$ satisfy
\begin{equation}
 \left|\Tr[Q\Phi(\rho_P)]-\frac{s}{n}\right|
 \geq
 \frac{f^{-1}(\lambda(\Phi))}{2}\sqrt{\frac{s}{r}}.
 \label{eq:born-witness}
\end{equation}
\end{corollary}

\begin{proof}
Choose $P,Q$ from \eqref{eq:converse} and write
$r=\rank(P)>0$, $s=\rank(Q)>0$. Multiplying that inequality by
$n$ and using $\tau(P)=r/n$, $\tau(Q)=s/n$ gives
\[
 |\Tr[Q\Phi(P)]-rs/n|
 \geq\frac{f^{-1}(\lambda(\Phi))}{2}\sqrt{rs}.
\]
The operator $\rho_P=P/r$ is positive with trace one. By linearity,
dividing the last inequality by $r$ gives \eqref{eq:born-witness}.
\end{proof}

The reference probability $s/n$ is the acceptance probability after
complete depolarization. Thus the converse produces an explicit type
of preparation and measurement test for a spectral obstruction to
mixing. There is also a direct dynamical consequence of the same spectral
parameter. If $\rho$ is a density operator and $t\geq1$ is an integer,
write $\Phi^t$ for the composition of $t$ copies of $\Phi$. Then
\begin{equation}
 \frac12\Tr\left|\Phi^t(\rho)-\frac{I}{n}\right|
 \leq \frac12\lambda(\Phi)^t
           \sqrt{n\Tr(\rho^2)-1}.
 \label{eq:trace-mixing}
\end{equation}
To see this, put $S=\Phi-\Pi$. The identities
$S\Pi=\Pi S=0$ and $\Pi^2=\Pi$ imply
$\Phi^t-\Pi=S^t$. Cauchy--Schwarz for the singular values gives
$\Tr|X|\leq\sqrt n\,\|X\|_{\mathrm F}$, and submultiplicativity
of the induced norm gives $\|S^t\|_{2\to2}\leq\lambda(\Phi)^t$.
Therefore
\[
 \frac12\Tr\left|\Phi^t(\rho)-\frac In\right|
 \leq\frac{\sqrt n}{2}\|S^t(\rho-I/n)\|_{\mathrm F}
 \leq\frac{\sqrt n}{2}\lambda(\Phi)^t\|\rho-I/n\|_{\mathrm F}.
\]
Finally, $\|\rho-I/n\|_{\mathrm F}^2=\Tr(\rho^2)-1/n$,
which proves \eqref{eq:trace-mixing}.
The left side is the trace distance to the maximally mixed state and
bounds the deviation of any measurement probability.
Together with \eqref{eq:inverse-delta}, this gives a mixing bound in
terms of projection discrepancy whenever the resulting upper bound
on $\lambda$ is below one.

\section{Sharp dimension bounds and explicit witnesses}
\label{sec:dimension}

\subsection{Sharp projection comparison and supplied witnesses}
\label{sec:projection-witnesses}

A spectral decomposition gives a finite list of tests with an
exact dimension constant. For integers $k\geq1$ and $m\geq1$, define
\begin{equation}
 d_k=\sqrt{k}-\sqrt{k-1},\qquad c_0=0,\quad
 c_m=\sum_{k=1}^m d_k^2,\qquad K_n=2c_{\lfloor n/2\rfloor}.
 \label{eq:dimension-constants}
\end{equation}
For $k\geq2$,
$1/(4k)\leq d_k^2\leq1/(4(k-1))$. Hence
$c_m=\frac14\log m+O(1)$ and $K_n=\frac12\log n+O(1)$.
The centered hypotheses below are exactly those of $\Phi-\Pi$.
The unrestricted comparison is proved separately in
Appendix~\ref{sec:unrestricted}.

\begin{theorem}[Sharp centered comparison]
\label{thm:centered-dimension}
If $S:\M\to\M$ is Hermiticity preserving and
$S(I)=S^*(I)=0$, then
\begin{equation}
 \|S\|_{\proj}\geq\frac{\|S\|_{2\to2}}{K_n}.
 \label{eq:centered-dimension-bound}
\end{equation}
The constant $K_n$ is optimal in every dimension.
For any traceless Hermitian $H$ with $S(H)\ne0$, a nonzero
projection pair $(P,Q)$ among at most $(n-1)^2$ complemented
spectral-prefix pairs of $H$ and
$S(H)$ satisfies
\begin{equation}
 \frac{|\Tr[Q S(P)]|}{\sqrt{\rank(P)\rank(Q)}}
 \geq\frac{\|S(H)\|_2}{K_n\|H\|_2}.
 \label{eq:centered-prefix-witness}
\end{equation}
\end{theorem}

\begin{proof}
Set $Y=S(H)$, which is traceless and Hermitian.
Order the eigenvalues of $H$ as $h_1\geq\cdots\geq h_n$.
Let $P_k$ project onto the first $k$ eigenvectors, and put
$g_k=h_k-h_{k+1}\geq0$ for $1\leq k<n$.
Since $\Tr(H)=0$,
\[
 H=\sum_{k=1}^{n-1}g_k(P_k-kI/n).
\]
Set $r_k=\min(k,n-k)$ and define
\[
 (R_k,\sigma_k)=
 \begin{cases}
 (P_k,1),&k\leq n/2,\\
 (I-P_k,-1),&k>n/2.
 \end{cases}
 \quad\text{and}\quad
 P_k-kI/n=\sigma_k(R_k-r_kI/n).
\]
Thus $\rank(R_k)=r_k$. The scalar terms disappear in a pairing
$\Tr[Y S(H)]$ on both sides, by $S(I)=S^*(I)=0$.

Define $f_0=f_n=0$, $f_k=\sqrt{\min(k,n-k)}$ for $0<k<n$,
and let $w=(w_1,\ldots,w_n)^\top$ have coordinates
$w_j=f_j-f_{j-1}$ for $1\leq j\leq n$. Summation by parts gives
\begin{equation}
 L_0(H):=\sum_{k=1}^{n-1}g_kf_k
        =\sum_{j=1}^n h_jw_j
        \leq\sqrt{K_n}\,\|H\|_{\mathrm F}.
 \label{eq:centered-cost}
\end{equation}
Indeed, if $m=\lfloor n/2\rfloor$, the vector $w$ is
\[
 (d_1,\ldots,d_m,-d_m,\ldots,-d_1)
 \quad\text{or}\quad
 (d_1,\ldots,d_m,0,-d_m,\ldots,-d_1)
\]
according as $n$ is even or odd. Thus
$\sum_jw_j=0$ and $\sum_jw_j^2=2c_m=K_n$.
For the ordered spectral prefixes $Q_l$ of $Y$, with $1\leq l<n$, define
$T_l=Q_l$, $\theta_l=1$ when $l\leq n/2$ and
$T_l=I-Q_l$, $\theta_l=-1$ otherwise. With
$s_l=\min(l,n-l)$ the corresponding identity is
$Q_l-lI/n=\theta_l(T_l-s_lI/n)$.
Write the eigenvalues of $Y$ as $y_1\geq\cdots\geq y_n$ and
put $b_l=y_l-y_{l+1}\geq0$ for $1\leq l<n$.
Let $M$ be the largest normalized absolute correlation over the
two lists of complemented prefixes. Expanding both arguments in
$\Tr[Y S(H)]$ and discarding the signs $\sigma_k\theta_l$ gives
\begin{align*}
 \|Y\|_{\mathrm F}^2
 \leq\sum_{k,l}g_kb_l|\Tr[T_lS(R_k)]|
 \leq M\left(\sum_k g_k\sqrt{r_k}\right)
          \left(\sum_l b_l\sqrt{s_l}\right)
          &=M L_0(H)L_0(Y)\\
& \leq M K_n\|H\|_{\mathrm F}\|Y\|_{\mathrm F}.
\end{align*}
Canceling $\|Y\|_{\mathrm F}$ proves
\eqref{eq:centered-prefix-witness}. There are at most $n-1$
prefixes on each side; repeated eigenvalues only create zero
coefficients that may be omitted.

For $S\ne0$, a Hermitian norm maximizer exists by
Lemma~\ref{lem:hermitian-attainment}. Removing its identity
component leaves its image unchanged and cannot increase its
norm, so a traceless maximizer can be chosen. This proves
\eqref{eq:centered-dimension-bound}; the zero map is immediate.

For optimality, use $D=\diag(w)$ and $T_D(X)=\Tr(DX)D$.
This map is centered because $\Tr(D)=0$, and
$\|T_D\|_{2\to2}=\Tr(D^2)=K_n$.
For a rank-$r$ projection $P$, both the positive contribution
and the absolute negative contribution to $\Tr(DP)$ are at
most $\sum_{j=1}^{\min(r,m)}d_j=\sqrt{\min(r,m)}$.
Their difference therefore has modulus at most $\sqrt r$,
with equality for the first coordinate projection. Since
$\Tr[QT_D(P)]=\Tr(DP)\Tr(DQ)$, the two independently optimized
projection factors give
\[
 \|T_D\|_{\proj}
 =\left(\max_{P\ne0\text{ projection}}
       \frac{|\Tr(DP)|}{\sqrt{\rank(P)}}\right)^2=1.
\]
This proves optimality.
\end{proof}

Combining the centered comparison with Theorem~\ref{thm:main}
gives, for any unital channel $\Phi$ and $\lambda=\lambda(\Phi)$,
\begin{equation}
 \max\left\{\frac{\lambda}{K_n},\frac{f^{-1}(\lambda)}2\right\}
 \leq\delta(\Phi)\leq\left(1-\frac1n\right)\lambda.
 \label{eq:combined-bounds}
\end{equation}
Algorithm~\ref{alg:supplied-witness} in the introduction implements
the complemented-prefix construction in this proof.
Proposition~\ref{prop:model-stability} gives its certificate when the
supplied channel model has bounded error.

For a classical channel $\Phi_A$ with uniform $n$-point marginals,
write $\delta=\delta(\Phi_A)$ and $\rho=\rho(A)$.
The constant $K_n$ is also obtained from
\citet[Theorem~2.2]{Peyre2013}.
Let $F_n:[0,1]\to[0,\infty)$ linearly interpolate the values
$F_n(k/n)=\sqrt{\min(k/n,1-k/n)}$ for $0\leq k\leq n$.
With the $w_j$ from the proof, on the $j$th interval,
\[
 F_n'=\sqrt n\,w_j\quad\text{and}\quad
 \int_0^1|F_n'(t)|^2\,dt=\sum_{j=1}^n w_j^2=K_n.
\]
All event probabilities lie on this grid. For input and output events
of probabilities $p$ and $q$, complementing either event bounds their
absolute covariance by $\delta F_n(p)F_n(q)$;
the interpolated values between grid points are never tested.
Taking both Sobolev factors to be $\sqrt\delta F_n$ in
Peyre's theorem gives $\rho\leq\delta K_n$.
Compression preserves $\rho=\lambda$ and does not increase
$\delta$. The direct matrix proof above additionally supplies
the finite prefix witnesses.

\begin{remark}[Centered maps and the depolarizing interior]
The centered-map inequality itself also follows from its channel
case. If $S(I)=S^*(I)=0$ and $S$ is Hermiticity preserving,
its Choi matrix, in the convention of \eqref{eq:choi-state},
is Hermitian, whereas
$J(\Pi)=I_{n^2}/n$ is positive definite. For nonzero $S$ and
$0<t\leq1/(n\|J(S)\|_{\mathrm{op}})$, the map
$\Phi_t=\Pi+tS$ is completely positive, unital, and trace
preserving. Since
$\lambda(\Phi_t)=t\|S\|_{2\to2}$ and
$\delta(\Phi_t)=t\|S\|_{\proj}$, the homogeneous channel
bound with constant $K_n$ transfers to $S$ by division by $t$.
This argument requires
both centering identities and does not remove the scale from
Peyre's nonhomogeneous logarithmic comparison.
\end{remark}

\subsection{Equality channels, rigidity, and one-use envelopes}
\label{sec:projection-envelopes}

Let $\mathcal Q_n$ be the set of unital quantum channels on $\M$.
For $0\leq\ell\leq1$, define the admissible set
\[
 \mathcal Q_n(\ell)=\{\Phi\in\mathcal Q_n:\lambda(\Phi)=\ell\}.
\]
The two attained lower envelopes over this set are
\begin{equation}
 d_n(\ell)=\min_{\Phi\in\mathcal Q_n(\ell)}\delta(\Phi)\quad\text{and}\quad
 e_n(\ell)=\min_{\Phi\in\mathcal Q_n(\ell)}\epsilon(\Phi).
 \label{eq:envelope-definitions}
\end{equation}
Attainment and equality with the corresponding classical minima follow
from Proposition~\ref{prop:classical-compression}.

\begin{theorem}[An exact channel lower envelope]
\label{thm:channel-envelope}
For every $n\geq2$ and every $0\leq\ell\leq K_n/n$,
\begin{equation}
 \min_{\Phi\in\mathcal Q_n(\ell)}
       \delta(\Phi)=\frac{\ell}{K_n}.
 \label{eq:channel-envelope}
\end{equation}
Thus $K_n$ is the optimal dimension constant for unital quantum
channels in every dimension.
\end{theorem}

\begin{proof}
Let $w$ be the vector from the preceding proof and $D=\diag(w)$.
Set
\[
 A_{ij}=\frac{1+w_iw_j}{n}\quad\text{and}\quad
 \Psi_n(X)=\Pi(X)+\frac{\Tr(DX)}nD.
\]
Because $\sum_iw_i=0$ and $|w_i|\leq1$, $A$ is nonnegative and
doubly stochastic. To verify the channel realization, its measure-and-prepare action is
\[
 \Phi_A(X)=\sum_i\left(\sum_j\frac{1+w_iw_j}{n}X_{jj}\right)
                   |i\rangle\langle i|
          =\Pi(X)+\frac{\Tr(DX)}nD=\Psi_n(X).
\]
Thus $\sqrt{A_{ij}}\,|i\rangle\langle j|$ are Kraus operators
for $\Psi_n$. Column sums of $A$ give trace preservation, and
row sums give unitality.
Its centered map is $T_D/n$, so
\begin{equation}
 \lambda(\Psi_n)=\frac{K_n}{n}\quad\text{and}\quad
 \delta(\Psi_n)=\frac1n.
 \label{eq:extremal-channel-values}
\end{equation}
For $t=n\ell/K_n\in[0,1]$, the unital channel
$\Phi_\ell=(1-t)\Pi+t\Psi_n$ has centered map
$t(\Psi_n-\Pi)$. Hence
$\lambda(\Phi_\ell)=\ell$ and $\delta(\Phi_\ell)=\ell/K_n$.
The reverse inequality for every channel follows from
Theorem~\ref{thm:centered-dimension} applied to $\Phi-\Pi$.
\end{proof}

\begin{lemma}[Rigidity of classical equality]
\label{lem:equality-rigidity}
Let $A$ be a nonnegative doubly stochastic $n\times n$ matrix with
$\lambda(\Phi_A)>0$ and
$\delta(\Phi_A)=\lambda(\Phi_A)/K_n=:d$.
Define $f_j=\sqrt{\min(j,n-j)}$ for $0\leq j\leq n$ and
$w=(w_1,\ldots,w_n)^\top$ with $w_j=f_j-f_{j-1}$ for $1\leq j\leq n$.
After independent permutations of its rows and columns,
\begin{equation}
 A=J_n/n+dww^\top.
 \label{eq:classical-equality-form}
\end{equation}
In particular, $d\leq1/n$ and $\lambda(\Phi_A)\leq K_n/n$.
\end{lemma}

\begin{proof}
Put $B=A-J_n/n$ and use ordinary Euclidean vector norms.
Choose a nonzero real mean-zero spectral
maximizer $x$, and set $y=Bx$. Independently permute the input
and output coordinates so that $x$ and $y$ are decreasing.
The corresponding row and column permutations of $A$ preserve
nonnegativity, both stochasticity conditions, the spectral norm,
and the discrepancy. Write
\[
 g_k=x_k-x_{k+1},\qquad h_l=y_l-y_{l+1},\qquad
 C_{lk}=\sum_{i\leq l,\,j\leq k}B_{ij}
 \quad(1\leq k,l<n).
\]
Since $B$ has zero row and column sums, either prefix in this
rectangle can be complemented without changing its absolute sum.
Equation~\eqref{eq:classical-parameters} therefore gives
$|C_{lk}|\leq d f_lf_k$. The centered prefix expansions and
summation by parts yield
\[
\|y\|_{\ell_2}^2
 =\sum_{k,l}g_kh_l C_{lk}
 \leq d\left(\sum_k g_kf_k\right)
          \left(\sum_l h_lf_l\right)
   =d(x\cdot w)(y\cdot w)
 \leq dK_n\|x\|_{\ell_2}\|y\|_{\ell_2}.
\]
The first and last quantities agree because
$\|y\|_{\ell_2}=\lambda(\Phi_A)\|x\|_{\ell_2}$ and
$dK_n=\lambda(\Phi_A)$. Equality in both Cauchy--Schwarz
bounds thus forces $x=cw$ and $y=ew$ for some $c,e>0$.

The coordinates of $w$ are strictly decreasing, including
across the central gap for even $n$ and the central zero for
odd $n$. Hence all $g_k,h_l$ are strictly positive. Since
\[
 \sum_{k,l}g_kh_l\bigl(df_kf_l-C_{lk}\bigr)=0
\]
and every summand is nonnegative, $C_{lk}=df_lf_k$ for all
$1\leq k,l<n$. The rectangle sums with either index $0$ or
$n$ vanish. Taking double differences gives
$B_{ij}=d(f_i-f_{i-1})(f_j-f_{j-1})=dw_iw_j$.
Finally $w_1=1$ and $w_n=-1$, so
$A_{1n}=1/n-d\geq0$. This proves both conclusions.
\end{proof}

\begin{corollary}[The exact extent of the linear envelope]
\label{cor:linear-envelope-endpoint}
For the envelope defined in \eqref{eq:envelope-definitions},
\begin{equation}
 \begin{cases}
 d_n(\ell)=\ell/K_n,&0\leq\ell\leq K_n/n,\\
 d_n(\ell)>\ell/K_n,&K_n/n<\ell\leq1.
 \end{cases}
 \label{eq:linear-envelope-extent}
\end{equation}
The second interval is empty when $n=2$.
\end{corollary}

\begin{proof}
The first branch is Theorem~\ref{thm:channel-envelope}.
Suppose $\lambda(\Phi)>0$ and
$\delta(\Phi)=\lambda(\Phi)/K_n$.
Proposition~\ref{prop:classical-compression} supplies $A$ with
the same spectral parameter and
$\delta(\Phi_A)\leq\delta(\Phi)$, whereas
Theorem~\ref{thm:centered-dimension} gives the reverse lower
bound $\delta(\Phi_A)\geq\lambda(\Phi)/K_n$.
Thus $\Phi_A$ is an equality case and
Lemma~\ref{lem:equality-rigidity} implies
$\lambda(\Phi)\leq K_n/n$.
At each fixed $\ell>K_n/n$, the minimum is attained by
Proposition~\ref{prop:classical-compression}; since equality
is impossible, that minimum is strictly larger than $\ell/K_n$.
\end{proof}

The rigidity lemma classifies a selected classical compression.
For example, all unital qubit
channels attain $\delta=\lambda/2$, including unitary channels.
Equation~\eqref{eq:linear-envelope-extent} identifies exactly
where the linear branch ends.
For odd $n$, the central zero in $w$ is essential for
tracelessness. The separate minima at $\lambda=1$ are evaluated
below.

\begin{corollary}[Necessity of the logarithmic order]
\label{prop:sharpness}
The channels $\Psi_n$ of Theorem~\ref{thm:channel-envelope}
have spectral parameters $\lambda_n:=\lambda(\Psi_n)=K_n/n\to0$
and satisfy $\delta(\Psi_n)=1/n$. Write $a_n\sim b_n$ when
$a_n/b_n\to1$ as $n\to\infty$. Then
\[
 \delta(\Psi_n)\sim
 \frac{2\lambda_n}{\log(1/\lambda_n)}.
\]
Consequently no dimension-independent linear lower bound
$\delta\geq c\lambda$ with $c>0$ holds for all unital channels.
\end{corollary}

\begin{proof}
Since $K_n\sim\frac12\log n$, we have
$\log K_n/\log n\to0$ and
$\log(1/\lambda_n)=\log n-\log K_n\sim\log n$.
Using the exact values in \eqref{eq:extremal-channel-values} gives
\[
 \frac{\delta(\Psi_n)}{2\lambda_n/\log(1/\lambda_n)}
 =\frac{\log n-\log K_n}{2K_n}\longrightarrow1.
\]
Moreover, $\delta(\Psi_n)/\lambda_n=1/K_n\to0$.
A universal inequality $\delta\geq c\lambda$ with fixed $c>0$
would require $1/K_n\geq c$ for every $n$, a contradiction.
\end{proof}

This proves order sharpness. The obstruction is already
classical. The harmonic rank-one construction of
\citet[Section~2.2, Theorem~5]{BollobasNikiforov2004},
also discussed by \citet[Section~1.3]{lev2015discrete}, is retained
as a numerical comparison in Section~\ref{sec:experiments}.

\begin{corollary}[Optimal forward coefficient and qubit channels]
\label{cor:qubit}
For every dimension $n\geq2$ and every $\ell\in[0,1]$, a unital
channel $\Phi$ attains
\[
 \lambda(\Phi)=\ell\quad\text{and}\quad
 \delta(\Phi)=\left(1-\frac1n\right)\ell.
\]
Thus the coefficient in the upper bound of
\eqref{eq:combined-bounds} is optimal. Every unital qubit channel
satisfies
\begin{equation}
 \delta(\Phi)=\frac{\lambda(\Phi)}2\quad\text{and}\quad
 \epsilon(\Phi)=\frac{\lambda(\Phi)}4.
 \label{eq:qubit-exact}
\end{equation}
\end{corollary}

\begin{proof}
For a unitary $U$, let $\mathcal U(X)=UXU^\dagger$.
Its restriction to the traceless subspace is an isometry, so
$\lambda(\mathcal U)=1$.
A rank-one projection $P$ and $Q=UPU^\dagger$ have
$|\Tr[Q(\mathcal U-\Pi)(P)]|=1-1/n$.
The upper bound proves equality, and
$\Phi=\ell\mathcal U+(1-\ell)\Pi$ gives the assertion for all
$\ell\in[0,1]$.
For $n=2$, $K_2=2$. The centered lower bound and the forward
upper bound both equal $\lambda/2$, proving the first qubit
identity. Every nonzero proper projection has normalized trace
$1/2$, while a projection equal to $0$ or $I$ gives zero centered
correlation. Consequently $\epsilon=\delta/2$, proving the
second identity.
\end{proof}

\subsection{Cut comparison and the separate endpoint minima}
\label{sec:cut-envelopes}

\begin{proposition}[Sharp finite-dimensional cut comparison]
\label{prop:sharp-cut}
Every unital channel $\Phi$ on $\M$ satisfies
\begin{equation}
 \frac{\lambda(\Phi)}{2n}\leq\epsilon(\Phi)
 \leq\frac{\lfloor n^2/4\rfloor}{n^2}\lambda(\Phi).
 \label{eq:sharp-cut-comparison}
\end{equation}
Both constants are optimal. More precisely,
\begin{equation}
 \min_{\Phi\in\mathcal Q_n(\ell)}\epsilon(\Phi)=\frac{\ell}{2n}
       \qquad(0\leq\ell\leq2/n).
 \label{eq:small-spectral-cut-envelope}
\end{equation}
\end{proposition}

\begin{proof}
For the upper bound, maximize
$\sqrt{p(1-p)q(1-q)}$ in the projection mixing estimate
over the rank grid $p,q\in\{0,1/n,\ldots,1\}$.
Its maximum is $\lfloor n^2/4\rfloor/n^2$.
For a unitary matrix $U$, let $\mathcal U(X)=UXU^\dagger$.
The channel $\ell\mathcal U+(1-\ell)\Pi$ attains this value times
$\ell$ at a rank-$\lfloor n/2\rfloor$ projection $P$
and $Q=UPU^\dagger$.

For the lower bound, write $\lambda=\lambda(\Phi)$ and
$\epsilon=\epsilon(\Phi)$. Suppose $\lambda>0$, put $S=\Phi-\Pi$,
and choose a traceless Hermitian maximizer $H$ and $Y=S(H)$.
For a Hermitian matrix $Z$, let $\operatorname{diam}(Z)$ be
its largest eigenvalue minus its smallest eigenvalue.
Use the uncomplemented spectral-prefix expansions from the
proof of Theorem~\ref{thm:centered-dimension}, with eigenvalue
gaps $g_k,h_l\geq0$. Each prefix correlation is at most
$n\epsilon$ in absolute value. Consequently
\[
\|Y\|_{\mathrm F}^2
 \leq n\epsilon\left(\sum_k g_k\right)
                         \left(\sum_l h_l\right)
 =n\epsilon\,\operatorname{diam}(H)
                   \operatorname{diam}(Y)
 \leq2n\epsilon\|H\|_{\mathrm F}\|Y\|_{\mathrm F}.
\]
The last inequality uses
$\operatorname{diam}(Z)\leq\sqrt2\|Z\|_{\mathrm F}$.
Canceling norms proves
$\lambda\leq2n\epsilon$; the case $\lambda=0$ is immediate.

For lower-bound equality, let
$D=\diag(1,0,\ldots,0,-1)$ and
\[
 \Phi_t(X)=\Pi(X)+t\Tr(DX)D,\qquad0\leq t\leq1/n.
\]
Its transition matrix has entries
$1/n+tD_{ii}D_{jj}$ and is nonnegative and doubly stochastic,
so this is an embedded classical channel. Since
$\Tr(D^2)=2$ and
$\max_{P\text{ projection}}|\Tr(DP)|=1$,
\[
 \lambda(\Phi_t)=2t\quad\text{and}\quad\epsilon(\Phi_t)=t/n.
\]
Taking $t=\ell/2$ proves
\eqref{eq:small-spectral-cut-envelope} and optimality of the
lower constant.
\end{proof}

The lower comparison also follows from
\citet[Theorem~2.2, p.~3]{Peyre2013}.
Let $A$ be the classical compression of a unital channel $\Phi$,
and write $\rho=\rho(A)$ and $\epsilon=\epsilon(\Phi_A)$.
Let $G_n:[0,1]\to[0,1]$ linearly interpolate $G_n(0)=G_n(1)=0$ and
$G_n(j/n)=1$ for $1\leq j<n$. Its slopes on the first and last
grid intervals are $n$ and $-n$, with zero slope between them,
so $\int_0^1|G_n'(t)|^2\,dt=2n$.
For input and output events of probabilities $p$ and $q$ under the
compressed classical law, the absolute covariance is bounded by
$\epsilon G_n(p)G_n(q)$. Applying Peyre's theorem with the
two functions $\sqrt\epsilon\,G_n$ gives $\rho\leq2n\epsilon$;
Proposition~\ref{prop:classical-compression} transfers this bound
to channels. The direct proof above supplies the finite-dimensional
spectral-diameter argument and explicit equality-channel realization.

The spectrum supported on two coordinates in this cut equality family
solves a
different optimization from the square-root-increment spectrum
for $\delta$. The endpoint minima below likewise need not
have a common minimizing channel.

\begin{corollary}[The exact extent of the linear cut envelope]
\label{cor:cut-envelope-endpoint}
The cut envelope in \eqref{eq:envelope-definitions} satisfies
\begin{equation}
 \begin{cases}
 e_n(\ell)=\ell/(2n),&0\leq\ell\leq2/n,\\
 e_n(\ell)>\ell/(2n),&2/n<\ell\leq1.
 \end{cases}
 \label{eq:cut-envelope-extent}
\end{equation}
The second interval is empty for $n=2$.
\end{corollary}

\begin{proof}
The first branch is Proposition~\ref{prop:sharp-cut}.
Suppose $\epsilon(\Phi)=\lambda(\Phi)/(2n)>0$ and $n\geq3$.
By Proposition~\ref{prop:classical-compression} there is a classical
compression $A$ with the same spectral parameter and no larger cut
discrepancy. Proposition~\ref{prop:sharp-cut} forces equality for
this compression as well. Write $\lambda=\lambda(\Phi_A)$,
$\epsilon=\epsilon(\Phi_A)=\lambda/(2n)$, and $B=A-J_n/n$.
Choose a nonzero real mean-zero norm maximizer $x$, put $y=Bx$,
and independently permute columns and rows so that $x,y$ are
decreasing. All vector norms below are ordinary Euclidean norms.
For a real vector $z$, write
$\operatorname{diam}(z)=\max_i z_i-\min_i z_i$,
and let $e_i$ denote the $i$th standard coordinate vector.
Set
\[
 g_k=x_k-x_{k+1},\qquad h_l=y_l-y_{l+1},\qquad
 C_{lk}=\sum_{i\leq l,\,j\leq k}B_{ij}
 \quad(1\leq k,l<n).
\]
The spectral-prefix proof gives
\[
 \|y\|_{\ell_2}^2
 =\sum_{k,l}g_kh_l C_{lk}
 \leq n\epsilon\,\operatorname{diam}(x)\operatorname{diam}(y)
 \leq2n\epsilon\|x\|_{\ell_2}\|y\|_{\ell_2}.
\]
The endpoints agree because $\|y\|_{\ell_2}=\lambda\|x\|_{\ell_2}$.
Equality in each diameter bound, obtained by pairing with $e_1-e_n$,
therefore forces $x=c(e_1-e_n)$ and $y=d(e_1-e_n)$ for $c,d>0$.
Their positive prefix gaps occur at $1$ and $n-1$.
Every deficit $n\epsilon-C_{lk}$ is nonnegative, so equality
in the weighted sum forces $C_{1,n-1}=n\epsilon=\lambda/2$.
Since the first row of $B$ sums to zero,
\[
 B_{1n}=-\lambda/2\quad\text{and}\quad A_{1n}=1/n-\lambda/2\geq0.
\]
Thus $\lambda\leq2/n$. At each fixed $\ell>2/n$ equality is
impossible, and attainment from
Proposition~\ref{prop:classical-compression} makes the minimum
strictly greater than $\ell/(2n)$.
\end{proof}

Zero middle gaps in the cutoff argument can leave other matrix
components unconstrained.

\begin{proposition}[Separate lower envelopes at spectral parameter one]
\label{prop:endpoint-envelopes}
For every $n\geq2$,
\begin{equation}
 \min_{\Phi\in\mathcal Q_n(1)}\delta(\Phi)
     =\frac{\lceil n/2\rceil}{n}\quad\text{and}\quad
 \min_{\Phi\in\mathcal Q_n(1)}\epsilon(\Phi)
     =\frac{n-1}{n^2}.
 \label{eq:endpoint-envelopes}
\end{equation}
If $\lambda(\Phi)=1$ and $\epsilon(\Phi)=(n-1)/n^2$, then
$\delta(\Phi)=1-1/n$. Thus the two minima cannot be attained
by the same channel for $n\geq4$; for $n=2,3$ they can.
\end{proposition}

\begin{proof}
Let $A$ be a classical compression at $\lambda=1$, and choose
real vectors $x,y\in\mathbb R^n$ whose scores
$X=x_{J_{\mathrm{in}}}$ and $Y=y_{J_{\mathrm{out}}}$
have mean zero and variance one under
$\Pr(J_{\mathrm{in}}=j,J_{\mathrm{out}}=i)=A_{ij}/n$,
with $\mathbb E[XY]=1$. Consequently
$\mathbb E(X-Y)^2=0$, so $X=Y$ almost surely.
This common variable is nonconstant. A nontrivial level event
has probability $p=r/n$, $1\leq r\leq n-1$, and its matching
input and output events have joint probability $p$. Their
projections therefore have centered correlation $p(1-p)$.
Complement both events if necessary to make
$p\leq\lfloor n/2\rfloor/n$. The normalized discrepancy
of this pair is $1-p\geq\lceil n/2\rceil/n$.
Also, every nontrivial grid value satisfies
$p(1-p)\geq(n-1)/n^2$. This proves both lower bounds.

For attainment, let $1\leq m\leq\lfloor n/2\rfloor$,
$l=n-m$, and let $R$ be a rank-$m$ projection. Consider
the measure-and-prepare channel
\[
 \mathcal F_m(X)=\frac{\Tr(RX)}m R
       +\frac{\Tr((I-R)X)}l(I-R).
\]
It is completely positive, trace preserving, and unital.
With $D=R-(m/n)I$ and $a=\Tr(D^2)=ml/n$,
\[
 (\mathcal F_m-\Pi)(X)=\frac{\Tr(DX)}aD.
\]
This is the Hilbert--Schmidt orthogonal projection onto
$\mathbb C D$, giving $\lambda(\mathcal F_m)=1$.
For every nonzero rank-$r$ projection $P$, the positive and absolute
negative contributions to $\Tr(DP)$ are at most
$(l/n)\min(r,m)$ and $(m/n)\min(r,l)$, respectively. Since
$\min(r,k)^2/r\leq k$ for every integer $k\geq1$, it follows that
\[
 \frac{|\Tr(DP)|^2}{r}
 \leq\max\left\{\frac{l^2m}{n^2},\frac{m^2l}{n^2}\right\}
 =\frac{l^2m}{n^2},
\]
with equality at $P=R$. Also $|\Tr(DP)|\leq a$, with
equality at $R$. Factoring the two projection traces yields
\[
 \delta(\mathcal F_m)=\frac{l}{n}\quad\text{and}\quad
 \epsilon(\mathcal F_m)=\frac{ml}{n^2}.
\]
Taking $m=\lfloor n/2\rfloor$ attains the $\delta$ minimum;
taking $m=1$ attains the $\epsilon$ minimum.

Finally, suppose $\epsilon(\Phi)=(n-1)/n^2$. The common
event above must satisfy $p(1-p)\leq\epsilon(\Phi)$.
The grid lower bound forces $p=1/n$ or $1-1/n$.
Complementing to probability $1/n$ gives
$\delta(\Phi)\geq1-1/n$, and the forward estimate gives
equality. For $n\geq4$, this exceeds $\lceil n/2\rceil/n$.
For $n=2,3$, the choice $m=1$ attains both minima.
\end{proof}

The common-score step is the finite common-data mechanism
underlying classical correlation extraction
\citep{Witsenhausen1975,GhaziKamathSudan2016}.
Proposition~\ref{prop:endpoint-envelopes} evaluates its
consequences on the uniform probability grid. The two bases
need not coincide, so this common data belongs to one input
and output experiment.

The known one-use envelope values and bounds are collected in
Table~\ref{tab:one-use-envelopes}. Both normalized ratios are
nondecreasing on $(0,1]$: for $0<a<b\leq1$,
\begin{equation}
 \frac{d_n(a)}a\leq\frac{d_n(b)}b\quad\text{and}\quad
 \frac{e_n(a)}a\leq\frac{e_n(b)}b.
 \label{eq:envelope-scaling}
\end{equation}
Indeed, for each envelope choose its own minimizer $\Phi_b$ at $b$.
The unital channel
$\Phi_a=(a/b)\Phi_b+(1-a/b)\Pi$ has centered map
$(a/b)(\Phi_b-\Pi)$, so $\lambda,\delta,\epsilon$ all scale by
$a/b$. Its discrepancy bounds the corresponding minimum at $a$.
In particular, scaling the respective minimizers in
Proposition~\ref{prop:endpoint-envelopes} gives the attainable bounds
\begin{equation}
 d_n(\ell)\leq\ell\frac{\lceil n/2\rceil}{n}\quad\text{and}\quad
 e_n(\ell)\leq\ell\frac{n-1}{n^2}
 \quad(0\leq\ell\leq1).
 \label{eq:envelope-attainable-upper}
\end{equation}
These are attainable upper bounds on the minimizing envelopes.

\begin{table}[t]
\centering\small
\caption{One-use lower envelopes over unital channels on $M_n(\mathbb C)$ at
fixed spectral parameter $0\leq\ell\leq1$. The lower bounds hold
throughout this interval; the upper bounds are attained by explicit
channels and therefore bound the minima. The function $f$ is defined
in \eqref{eq:intro-main}. Strict-cutoff intervals are empty for $n=2$.}
\label{tab:one-use-envelopes}
\begin{tabular}{@{}p{0.24\textwidth}p{0.35\textwidth}p{0.35\textwidth}@{}}
\hline
 & Projection envelope $d_n(\ell)$ & Cut envelope $e_n(\ell)$\\
\hline
At $\ell=0$ & $0$ & $0$\\[3pt]
Exact initial branch & $\ell/K_n$ for $0\leq\ell\leq K_n/n$
                     & $\ell/(2n)$ for $0\leq\ell\leq2/n$\\[3pt]
Above the linear branch & $>\ell/K_n$ for $\ell>K_n/n$
                        & $>\ell/(2n)$ for $\ell>2/n$\\[3pt]
At $\ell=1$ & $\lceil n/2\rceil/n$ & $(n-1)/n^2$\\[3pt]
Lower bound & $\max\{\ell/K_n,f^{-1}(\ell)/2\}$ & $\ell/(2n)$\\[3pt]
Attainable upper bound & $\ell\lceil n/2\rceil/n$ & $\ell(n-1)/n^2$\\
\hline
\end{tabular}
\end{table}

\subsection{Stability of supplied witnesses}
\label{sec:witness-stability}

\begin{proposition}[Stability and a supplied-model certificate]
\label{prop:model-stability}
Let $\eta\geq0$ and let $\Phi,\Psi$ be unital channels on $\M$
with a certified bound $\|\Phi-\Psi\|_{2\to2}\leq\eta$. Then
\begin{align}
 |\lambda(\Phi)-\lambda(\Psi)|&\leq\eta, &
 |\zeta(\Phi)-\zeta(\Psi)|&\leq\eta,\notag\\
 |\delta(\Phi)-\delta(\Psi)|&\leq(1-1/n)\eta, &\text{and}\quad 
 |\epsilon(\Phi)-\epsilon(\Psi)|&\leq\frac{\lfloor n^2/4\rfloor}{n^2}\eta.
 \label{eq:parameter-stability}
\end{align}
Suppose $T=\Psi-\Pi$, $\lambda(\Psi)>0$, and a supplied
nonzero traceless Hermitian $H$ satisfies
$\|T(H)\|_2\geq\gamma\lambda(\Psi)\|H\|_2$ for
$0<\gamma\leq1$. Choose a pair $P,Q$ maximizing the model
score among the at most $(n-1)^2$ complemented prefix pairs
of $H,T(H)$ in Theorem~\ref{thm:centered-dimension}.
For the fixed pair and any unital channel $\Omega$ on $\M$, define
\[
 d_\Omega(P,Q)=
 \frac{|\Tr[Q(\Omega-\Pi)(P)]|}{\sqrt{\rank(P)\rank(Q)}}.
\]
It obeys the model bound
$d_\Psi(P,Q)\geq\gamma\lambda(\Psi)/K_n$ and the
true-channel bounds
\begin{equation}
 d_\Phi(P,Q)
 \geq\frac{\gamma\lambda(\Psi)}{K_n}-(1-1/n)\eta
 \geq\frac{\gamma\lambda(\Phi)}{K_n}
       -\left(\frac\gamma{K_n}+1-\frac1n\right)\eta.
 \label{eq:model-certificate}
\end{equation}
Each lower bound may be replaced by its positive part.
\end{proposition}

\begin{proof}
The difference $\Delta=\Phi-\Psi$ is centered on both sides.
Thus for $p=\tau(P)$ and $q=\tau(Q)$,
\[
 |\tau[Q\Delta(P)]|
 \leq\eta\sqrt{p(1-p)q(1-q)}.
\]
The reverse triangle inequality for scalar absolute values therefore gives,
for every nonzero projection pair,
\[
 |d_\Phi(P,Q)-d_\Psi(P,Q)|
 \leq\eta\sqrt{(1-p)(1-q)}\leq(1-1/n)\eta.
\]
For proper projections, dividing the centered pairing bound by
$\sqrt{p(1-p)q(1-q)}$ gives the uniform error $\eta$ for the
$\zeta$ scores. Leaving it unnormalized and maximizing over the
rank grid gives $\eta\lfloor n^2/4\rfloor/n^2$ for the cut scores.
The difference of two maxima is at most the largest pointwise
difference, proving the three discrepancy bounds. The $\lambda$
bound is the reverse triangle inequality for the induced map norm.
For the supplied input,
$T(H)\ne0$, so Theorem~\ref{thm:centered-dimension} gives
the stated model score. The same pair on $\Phi$ loses at most
$\eta\sqrt{(1-p)(1-q)}\leq(1-1/n)\eta$.
Finally use $\lambda(\Psi)\geq\lambda(\Phi)-\eta$.
\end{proof}

\section{Parallel channel tests}
\label{sec:parallel}

Independent copies can increase rank-normalized projection
discrepancy. The construction is
classical after spectral compression, and its locality is across
channel uses. First we separate this question from sequential mixing.
For an integer $k\geq1$, parallel uses mean $\Phi^{\otimes k}$,
with reference
$\Pi_{n^k}=\Pi^{\otimes k}$; they are distinct from the sequential
iterate $\Phi^k$.

\begin{example}[Identical one-use parameters, different sequential mixing]
\label{ex:shift-refresh}
On four basis labels $(a,b)\in\{0,1\}^2$, the classical
channel $(a,b)\mapsto(b,c)$ refreshes $c$ as an independent
fair bit. In the order $00,01,10,11$, its doubly stochastic
matrix is
\[
 A=\frac14(J_4+yx^\top),\qquad
 x=(1,-1,1,-1)^\top,\quad y=(1,1,-1,-1)^\top.
\]
Both vectors have zero sum and ordinary Euclidean norm $2$,
and $x^\top y=0$.
Consequently $\|A-J_4/4\|_{\ell_2\to\ell_2}=1$ and $A^2=J_4/4$; the embedded
quantum channel satisfies $\lambda(\Phi_A)=1$ and
$\Phi_A^2=\Pi$ on the whole matrix algebra.
For $D_x=\diag(x)$ and $D_y=\diag(y)$,
\[
 (\Phi_A-\Pi)(X)=\frac14\Tr(D_xX)D_y.
\]
For either diagonal matrix, a nonzero rank-$r$ projection has absolute
trace pairing at most $\min(r,4-r)$. The maximum after division
by $\sqrt r$ is $\sqrt2$, attained at rank two. Factoring the
two pairings gives
$(\lambda,\delta,\epsilon)=(1,1/2,1/4)$.
The balanced block channel $\mathcal F_2$ from
Proposition~\ref{prop:endpoint-envelopes}, with
$R=(I+D_y)/2$, has the same tuple but satisfies
$\mathcal F_2^2=\mathcal F_2$. Thus these one-use parameters
do not determine sequential persistence. The common input
second bit and output first bit are randomized away on the
next use of $\Phi_A$.
\end{example}

For parallel uses, define
\[
 \delta_k=\delta(\Phi^{\otimes k})\quad\text{and}\quad
 \epsilon_k=\epsilon(\Phi^{\otimes k}).
\]

\subsection{The finite-copy endpoint}

\begin{lemma}[Discrepancy of a block-averaging channel]
\label{lem:block-averaging-discrepancy}
Let $N,s$ be positive integers. Consider nonzero pairwise orthogonal
projections $R_1,\ldots,R_s$ in $M_N(\mathbb C)$ that sum to $I_N$.
Define the block-averaging channel $\mathcal F$ on $M_N(\mathbb C)$,
the block proportions $p_i$, and their minimum $a$ by
\[
 \mathcal F(X)=\sum_{i=1}^s
       \frac{\Tr(R_iX)}{\Tr(R_i)}R_i,\qquad
 p_i=\frac{\Tr(R_i)}N,\qquad a=\min_i p_i.
\]
Then $\delta(\mathcal F)=1-a$.
\end{lemma}

\begin{proof}
The map $\mathcal F$ is the Hilbert--Schmidt orthogonal
projection onto the span of the $R_i$. Thus
$S=\mathcal F-\Pi_N$ is also an orthogonal projection.
With $\tau_N(X)=\Tr(X)/N$, for nonzero projections $P,Q$
the associated positive semidefinite form satisfies
\[
 |\tau_N[Q S(P)]|
 \leq\sqrt{\tau_N[P S(P)]\,\tau_N[Q S(Q)]}.
\]
After division by $\sqrt{\tau_N(P)\tau_N(Q)}$, this shows
that the maximum is attained with $P=Q$.
For a nonzero projection $P$, let $z_i=\tau_N(R_iP)$ and
$t=\tau_N(P)$. Then $0\leq z_i\leq p_i$,
$\sum_i z_i=t$, and its diagonal score equals
\[
 \frac{\tau_N[P S(P)]}{t}
     =\frac1t\sum_i\frac{z_i^2}{p_i}-t.
\]
If $t\geq a$, the inequalities $z_i^2/p_i\leq z_i$ bound
this by $1-t\leq1-a$. If $t\leq a$, then
\[
 \sum_i\frac{z_i^2}{p_i}\leq\frac1a\sum_i z_i^2
       \leq\frac{t^2}{a},
\]
so the score is at most $t(1/a-1)\leq1-a$.
A smallest block projection $P=Q=R_i$ attains $1-a$.
\end{proof}

\begin{corollary}[The finite-copy envelope at spectral parameter one]
\label{cor:finite-copy-endpoint}
For every $n\geq2$ and integer $k\geq1$,
\begin{equation}
 \min_{\Phi\in\mathcal Q_n(1)}
       \delta(\Phi^{\otimes k})
       =1-\left(\frac{\lfloor n/2\rfloor}{n}\right)^k.
 \label{eq:finite-copy-endpoint}
\end{equation}
The balanced two-block channel
$\mathcal F_{\lfloor n/2\rfloor}$ from
Proposition~\ref{prop:endpoint-envelopes} attains the minimum
simultaneously for every $k$.
\end{corollary}

\begin{proof}
The common-score argument in
Proposition~\ref{prop:endpoint-envelopes} supplies matching
input and output events of marginal and joint probability
$p=r/n$, with $0<p\leq\lfloor n/2\rfloor/n$.
Let $P,Q$ be the lifted projections for these events, so that
$\tau(P)=\tau(Q)=\tau[Q\Phi(P)]=p$.
Their $k$-fold intersections are represented by $P^{\otimes k}$ and
$Q^{\otimes k}$. Multiplicativity of normalized trace gives
\[
 \tau_{n^k}(P^{\otimes k})=\tau_{n^k}(Q^{\otimes k})=p^k\quad\text{and}\quad
 \tau_{n^k}[Q^{\otimes k}\Phi^{\otimes k}(P^{\otimes k})]=p^k.
\]
Their discrepancy is
\[
 \frac{p^k-p^{2k}}{p^k}=1-p^k
     \geq1-\left(\frac{\lfloor n/2\rfloor}{n}\right)^k.
\]
For attainment put $m=\lfloor n/2\rfloor$.
The map $\mathcal F_m^{\otimes k}$ averages over the
$2^k$ orthogonal tensor-product blocks. Its smallest block
has dimension $m^k$ in dimension $n^k$, so
Lemma~\ref{lem:block-averaging-discrepancy} gives precisely
$1-(m/n)^k$. The original channel has spectral parameter one
by Proposition~\ref{prop:endpoint-envelopes}.
\end{proof}

The feasible class consists of tensor powers of a single channel on
$M_n(\mathbb C)$. For $n=2$ and
$k=2$ this minimum is $3/4$, whereas the unrestricted dimension-four
minimum at spectral parameter one is $1/2$ by
Proposition~\ref{prop:endpoint-envelopes}. The finite-copy result also
implies $\delta_k\to1$ whenever $\lambda(\Phi)=1$.

\subsection{Tensorization and Gaussian-threshold witnesses}

The following theorem transfers classical correlation extraction
to projection tests. Its Gaussian threshold step is
Witsenhausen's rounding construction
\citep{Witsenhausen1975}; see also
\citet[Theorem~2.19]{GhaziKamathSudan2016}.
Its spectral
tensorization is the tracial channel case of quantum
maximal-correlation tensorization \citep[Theorem~2(a)]{Beigi2013}, through
the Choi dictionary \eqref{eq:choi-maximal}.
We supply direct proofs using the Hermitian maximizers justified above.

\begin{theorem}[Regularized projection and cut discrepancy]
\label{thm:regularized}
Fix a unital channel $\Phi:\M\to\M$ and write
$\lambda=\lambda(\Phi)$. For every $k\geq1$,
\begin{equation}
 \lambda(\Phi^{\otimes k})=\lambda.
 \label{eq:tensor-lambda}
\end{equation}
Both $\delta_k$ and $\epsilon_k$ are nondecreasing and have limits
$\delta_\infty$ and $\epsilon_\infty$. For $0<\lambda<1$,
\begin{equation}
\max\left\{\delta(\Phi),\frac{\arcsin\lambda}{\pi}\right\}
   \leq\delta_\infty\leq\lambda \quad\text{and}\quad \max\left\{\epsilon(\Phi),\frac{\arcsin\lambda}{2\pi}\right\}
   \leq\epsilon_\infty\leq\frac{\lambda}{4}.
 \label{eq:regularized-delta}
\end{equation}
If $\lambda=0$, both limits vanish. If $\lambda=1$, then
$\delta_\infty=1$ and $\epsilon_\infty=1/4$.
The exact finite-copy minimum at $\lambda=1$ is given in
Corollary~\ref{cor:finite-copy-endpoint}.
For $0<\lambda<1$, the arcsine lower bounds are attained as
limits of witness values using product-basis-diagonal
projections whose normalized ranks tend to $1/2$.
\end{theorem}

\begin{proof}
On normalized Hilbert--Schmidt space, let
$T=\Phi|_{I^\perp}$ be the restriction to the traceless subspace.
Then $\M=\mathbb CI\oplus I^\perp$ and $\Phi=1\oplus T$,
with $\|T\|_{2\to2}=\lambda$. Tensoring this orthogonal decomposition,
a sector with $1\leq r\leq k$ traceless factors has operator norm
$\lambda^r$. Subtracting $\Pi_{n^k}$ removes the all-identity
sector. Since $0\leq\lambda\leq1$, the maximum remaining norm
is $\lambda$, proving \eqref{eq:tensor-lambda}.
This uses Hilbert-space tensor norms and does not require
normality of $\Phi$.

Padding $P,Q$ by an identity factor preserves normalized ranks
and correlations, so both optimized discrepancies are
nondecreasing. The forward estimate gives
\[
 \delta_k\leq(1-n^{-k})\lambda\quad\text{and}\quad
 \epsilon_k\leq\lambda/4.
\]
For normalized projection ranks $p,q\in[0,1]$, the second bound follows
from $\sqrt{p(1-p)q(1-q)}\leq1/4$.
Hence both limits exist.

Suppose $0<\lambda<1$ and put $S=\Phi-\Pi$.
Choose $H$ as in Proposition~\ref{prop:classical-compression},
normalized by $\|H\|_2=1$.
Let $U,V$ be unitary eigenbasis matrices and let
$x,y\in\mathbb R^n$ satisfy
\[
 H=U\diag(x)U^\dagger\quad\text{and}\quad
 S(H)=\lambda V\diag(y)V^\dagger.
\]
Let $A$ be the doubly stochastic compression associated with these
bases in \eqref{eq:stochastic-compression}.
The vectors $x,y$ have zero uniform means and unit uniform
variances, and $Ax=\lambda y$.
Under the finite joint law
\[
 \Pr(J_{\mathrm{in}}=j,J_{\mathrm{out}}=i)=A_{ij}/n,
\]
the random variables $X=x_{J_{\mathrm{in}}}$ and $Y=y_{J_{\mathrm{out}}}$ satisfy
\begin{equation}
 \mathbb EX=\mathbb EY=0,\quad
 \mathbb EX^2=\mathbb EY^2=1,\quad
 \mathbb E[XY]=n^{-1}y^\top Ax=\lambda.
 \label{eq:pair-moments}
\end{equation}

Let $(X_r,Y_r)$, $1\leq r\leq k$, be independent copies of $(X,Y)$.
Let $P_k$ select the input strings $(j_1,\ldots,j_k)\in[n]^k$
with $\sum_{r=1}^k x_{j_r}\geq0$ in the basis $U^{\otimes k}$,
and let $Q_k$ select output strings $(i_1,\ldots,i_k)\in[n]^k$ with
$\sum_{r=1}^k y_{i_r}\geq0$ in the basis $V^{\otimes k}$.
The corresponding transition matrix is $A^{\otimes k}$.
Writing $\tau_{n^k}=\Tr/n^k$, we therefore have
\begin{align*}
 p_k:=\tau_{n^k}(P_k)&=\Pr\left(\sum_r X_r\geq0\right),\\
 q_k:=\tau_{n^k}(Q_k)&=\Pr\left(\sum_r Y_r\geq0\right),\\
 t_k:=\tau_{n^k}[Q_k\Phi^{\otimes k}(P_k)]
   &=\Pr\left(\sum_r X_r\geq0,\ \sum_r Y_r\geq0\right).
\end{align*}
Both projections are nonzero and proper: a finite mean-zero
unit-variance variable has positive and negative values.

The bivariate central limit theorem gives convergence of the
normalized sums to a centered Gaussian pair $(G,Z)$ with
unit variances and correlation $\lambda$. Its coordinate axes
have probability zero, so
\begin{equation}
 p_k,q_k\longrightarrow\frac12\quad\text{and}\quad
 t_k\longrightarrow
 \frac14+\frac{\arcsin\lambda}{2\pi}.
 \label{eq:gaussian-threshold}
\end{equation}
Let $W$ be a standard normal variable independent of $G$.
Put $\alpha=\arccos\lambda$ and represent
$Z=\cos\alpha\,G+\sin\alpha\,W$.
The half-planes $G\geq0$ and $Z\geq0$ intersect in a wedge of
angle $\pi-\alpha$. The rotationally uniform angle of $(G,W)$
therefore gives
$\Pr(G\geq0,Z\geq0)=(\pi-\alpha)/(2\pi)
=1/4+\arcsin(\lambda)/(2\pi)$.
Possible finite-sample ties at zero do not affect convergence:
the limiting marginals have no atom there. Consequently
\[
 t_k-p_kq_k\longrightarrow\frac{\arcsin\lambda}{2\pi}\quad\text{and}\quad
 \frac{|t_k-p_kq_k|}{\sqrt{p_kq_k}}
       \longrightarrow\frac{\arcsin\lambda}{\pi}.
\]
Each witness value is bounded by the corresponding optimized
parameter, proving both lower bounds in \eqref{eq:regularized-delta}.

If $\lambda=0$, $\Phi=\Pi$, so all discrepancies vanish.
If $\lambda=1$, Corollary~\ref{cor:finite-copy-endpoint} gives
$\delta_\infty=1$. For the cut discrepancy, the construction of
\eqref{eq:pair-moments} still applies and
$\mathbb E(X-Y)^2=0$, so $X=Y$ almost surely. Threshold sums
of this common mean-zero unit-variance variable give a balanced
limiting witness: the univariate central limit theorem yields
$p_k=q_k\to1/2$ and $t_k=p_k$. Hence the cut correlation tends
to $1/4$, proving $\epsilon_\infty=1/4$ without a nondegenerate
bivariate limit at $\lambda=1$.
\end{proof}

\begin{figure}[t]
\centering
\begin{minipage}{0.94\linewidth}
\centering
\[
 \Phi
 \xrightarrow{\text{two spectral bases}} A
 \xrightarrow{\Pr(J_{\mathrm{in}}=j,J_{\mathrm{out}}=i)=A_{ij}/n} (X,Y)
\]
\begin{tabular}{@{}c@{\qquad\qquad}c@{}}
$\swarrow$ & $\searrow$\\
finite complemented prefixes & independent copies and thresholds\\[2pt]
$\delta\geq\lambda/K_n$ &
$\delta_\infty\geq\arcsin(\lambda)/\pi$
\end{tabular}
\end{minipage}
\caption{Two projection-test constructions from a classical
compression. The second transfers classical Gaussian rounding to
product-basis tests.}
\label{fig:mechanism}
\end{figure}

The input state $P_k/\rank(P_k)$ is the uniform conditional
mixture of the accepted product basis states. It is separable,
although its labels are generally classically correlated across
uses. Measuring each output in the basis $V$ and thresholding the
sum of the observed $y$-values realizes the test $Q_k$.
If the bases and coordinates are available for implementation,
rejection sampling from uniform input labels accepts with
probability $p_k\to1/2$. The actual Born-probability difference
from complete depolarization is
\[
 \left|\frac{t_k}{p_k}-q_k\right|
   \longrightarrow\frac{\arcsin\lambda}{\pi}
 \qquad(0<\lambda<1).
\]
This requires neither entanglement across channel uses nor an
output measurement coupling different uses. 

The limits are taken for each fixed finite channel. The endpoint $\lambda=1$
has the exact finite-copy envelope in
Corollary~\ref{cor:finite-copy-endpoint}.
The Choi states with $\lambda<1$ belong to the resource class with
maximally mixed marginals and maximal correlation below one studied by
\citet{QinYao2021}, who obtain finite-copy approximation
results for nonlocal-game values. Transferring such additive approximation
bounds to rank-normalized discrepancy requires control of small event
marginals.

Only the optimized $\delta_k,\epsilon_k$ are
monotone; the values of the particular threshold sequence need
not be monotone or improve the one-use optimum. For example,
a binary symmetric channel with correlation $\lambda\in(0,1)$
has one-use discrepancy $\lambda/2$, which is larger than the
threshold limit $\arcsin(\lambda)/\pi$.
This is compatible with \eqref{eq:regularized-delta}, since
padding retains the better one-use witness. In fact optimized binary
projection discrepancy already increases with two uses, as follows.

\begin{proposition}[Exact binary amplification in two uses]
\label{prop:binary-two-copy}
For $0\leq\rho\leq1$, let $\mathcal B_\rho=\Phi_{A_\rho}$
be the embedded binary symmetric channel with
\[
 A_\rho=\frac12
 \begin{pmatrix}1+\rho&1-\rho\\1-\rho&1+\rho\end{pmatrix}.
\]

Then
\begin{equation}
 \delta(\mathcal B_\rho)=\frac\rho2,\qquad
 \delta(\mathcal B_\rho^{\otimes2})
       =\frac\rho2+\frac{\rho^2}{4},\qquad
 \epsilon(\mathcal B_\rho^{\otimes k})=\frac\rho4
       \quad(k\geq1).
 \label{eq:binary-exact-discrepancies}
\end{equation}
\end{proposition}

\begin{proof}
The one-use identities follow from
$\lambda(\mathcal B_\rho)=\rho$ and
Corollary~\ref{cor:qubit}.
In $A_\rho^{\otimes2}-J_4/4$, the entries at Hamming
distances $0,1,2$ are, respectively,
\[
 a=\frac{2\rho+\rho^2}{4},\qquad
 b=-\frac{\rho^2}{4},\qquad
 c=\frac{-2\rho+\rho^2}{4}.
\]
Centering makes full-rank rectangles zero and permits each
rank-three subset to be complemented to rank one with the same
absolute numerator and smaller denominator. It therefore
suffices to consider ranks one and two. Enumerating the two
geometric types of rank-two subsets (adjacent and antipodal)
gives the largest normalized absolute rectangle sums
\[
 \begin{array}{c|c}
 \text{ranks}&\text{largest normalized absolute sum}\\\hline
 (1,1)&(2\rho+\rho^2)/4\\
 (1,2)\text{ or }(2,1)&\rho/(2\sqrt2)\\
 (2,2)&\rho/2.
 \end{array}
\]
For the middle row, the possible two-entry sums are
$a+b=\rho/2$, $a+c=\rho^2/2$, $2b=-\rho^2/2$,
and $b+c=-\rho/2$, so the largest absolute sum is $\rho/2$.
For the last row, two adjacent subsets aligned along the same
coordinate give absolute sum $\rho$; adjacent subsets along
different coordinates give zero. Two antipodal subsets give
absolute sum $\rho^2$, and an adjacent--antipodal pair gives zero.
Division by two therefore gives $\rho/2$ as the maximum. The first displayed
quantity dominates the others on $0\leq\rho\leq1$ and is
attained by matching singleton subsets.
Equation~\eqref{eq:classical-parameters}, applied in dimension
four, identifies this with the optimum over all quantum
projections.
Finally identity padding gives
$\epsilon(\mathcal B_\rho^{\otimes k})\geq\rho/4$,
while Theorem~\ref{thm:regularized} and its finite-copy upper
bound give the reverse inequality.
\end{proof}

More generally, Corollary~\ref{cor:qubit}, identity padding, and
Theorem~\ref{thm:regularized} imply
$\epsilon(\Phi^{\otimes k})=\lambda(\Phi)/4$ for every unital
qubit channel and every $k\geq1$.
For $\Phi=\mathcal B_\rho$ at $\rho=0.6$, the optimized discrepancies are
$\delta_1=0.3$ and $\delta_2=0.39$, whereas
$\epsilon_k=0.15$ for every $k$. This analytical example separates
weaker selected balanced thresholds from optimized projection discrepancy,
which already increases at two uses, while optimized cut discrepancy
remains fixed.

For the extremizing channels in
Theorem~\ref{thm:channel-envelope}, the one-use discrepancy is
$1/n$ and $\lambda=K_n/n$. Hence
\begin{equation}
 \frac{\delta_\infty(\Psi_n)}{\delta(\Psi_n)}
 \geq \frac{n\,\arcsin(K_n/n)}{\pi}
 \sim\frac{\log n}{2\pi}.
 \label{eq:collective-gain}
\end{equation}
The regularization limit is taken first for each $n$.
Thus the possible gain over the one-use optimum is unbounded
as dimension grows, even though the construction uses only
separable preparations and local measurements.
This is a gain in rank-normalized correlation.
Section~\ref{sec:experiments} compares it with repeated one-use
tests at the same channel-use budget.

\section{Comparisons and numerical illustrations}
\label{sec:examples}
\label{sec:experiments}

The following examples distinguish ambient cut normalization,
measurement-axis restrictions, and collective threshold behavior.
The numerical studies use synthetic channels and simulated
observations; they illustrate the proved mechanisms.

\subsection{Relation to quantum cut uniformity}

A superoperator $\Phi:\M\to\M$ is irreducibly covariant if there
are a compact group $G$ and continuous irreducible unitary
representations $U,V$ on $\mathbb C^n$ satisfying the following
identity. For $g\in G$, write $U_g=U(g)$ and $V_g=V(g)$:
\begin{equation}
 \Phi(U_g XU_g^\dagger)=V_g\Phi(X)V_g^\dagger
 \quad\text{for all }g\in G,\ X\in\M.
 \label{eq:covariance}
\end{equation}
For this class, \citet[Corollary~3.4]{bannink2020quasirandom} prove
\[
 \lambda(\Phi)\leq2\pi^2\epsilon(\Phi).
\]
Our theorem instead compares $\lambda$ with $\delta$ for all
unital quantum channels. In the common setting of irreducibly
covariant unital channels, the preceding result gives a
dimension-independent linear bound
$\lambda(\Phi)\leq2\pi^2\delta(\Phi)$. The logarithmic converse here
addresses the class without this symmetry assumption.

The forward estimate also gives the somewhat sharper normalization
check
\begin{equation}
 \epsilon(\Phi)\leq\frac{\lambda(\Phi)}4,
 \label{eq:cut-quarter}
\end{equation}
since the normalized projection ranks $p=\tau(P)$ and $q=\tau(Q)$
satisfy $p(1-p),q(1-q)\leq1/4$. Proposition~\ref{prop:sharp-cut}
sharpens the coefficient to $\lfloor n^2/4\rfloor/n^2$, which
is strictly smaller for odd $n$, and gives the optimal lower
coefficient $1/(2n)$. Both comparisons use unitality and trace preservation.

\begin{example}[Small cut discrepancy with a persistent spectral mode]
\label{ex:cut-separation}
Fix $0<t\leq1$, let $R=|1\rangle\langle1|$ and
$R^\perp=I-R$, and define
\[
 \mathcal L_n(X)=\Tr(RX)R+
       \frac{\Tr(R^\perp X)}{n-1}R^\perp\quad\text{and}\quad
 \Phi_{n,t}=t\mathcal L_n+(1-t)\Pi.
\]
This is a unital measure-and-prepare channel. Set $H=R-I/n$ and
$a=\Tr(H^2)=(n-1)/n$. Direct calculation gives
\[
 (\mathcal L_n-\Pi)(X)=\frac{\Tr(HX)}a H.
\]
Thus $\mathcal L_n-\Pi$ is a Hilbert--Schmidt orthogonal
projection onto $\mathbb CH$, and
\begin{equation}
 \lambda(\Phi_{n,t})=t,\qquad
 \epsilon(\Phi_{n,t})=\frac{t(n-1)}{n^2},\qquad
 \delta(\Phi_{n,t})=\frac{t(n-1)}n.
 \label{eq:separation-values}
\end{equation}
To verify the last two equalities, the positive eigenvalue of $H$
is $a$ and the sum of its negative eigenvalues is $-a$. Consequently
$|\Tr(HP)|\leq a$ for every projection $P$, with equality at $R$.
Since nonzero projections have rank at least one,
$|\Tr(HP)|/\sqrt{\rank(P)}\leq a$, again with equality at $R$.
The rank-one formula for $\mathcal L_n-\Pi$ now gives
\eqref{eq:separation-values}.

For fixed $t>0$, the cut discrepancy tends to zero as $n$ grows,
while the spectral parameter and the rank-normalized discrepancy
stay bounded away from zero. The preparation $R$ and measurement
$\{R,I-R\}$ have acceptance bias $t(1-1/n)$ relative to complete
depolarization. Thus a mode localized on a small subspace is
visible to $\delta$ even when its contribution to $\epsilon$
vanishes.
\end{example}

\subsection{Spectral profiles and the sharp dimension constant}
\label{sec:spectral-profiles}

Let $m=\lfloor n/2\rfloor$ and choose
$1\geq w_1\geq\cdots\geq w_m\geq0$, with $w_1>0$.
Let $v\in\mathbb R^n$ contain these weights and their negatives, together with
one zero when $n$ is odd, and write $D=\diag(v)$. For
$0\leq\beta\leq1$, the matrix
$A_{ij}=(1+\beta v_iv_j)/n$ is nonnegative and doubly stochastic.
Its associated channel is
\[
 \Phi_{v,\beta}(X)=\Pi(X)+\frac{\beta}{n}\Tr(DX)D.
\]
For $1\leq r\leq m$, define $W_r=\sum_{j=1}^r w_j$.
The channel parameters satisfy the exact formulas
\begin{equation}
 \lambda(\Phi_{v,\beta})=\frac{2\beta}{n}\sum_{j=1}^m w_j^2\quad\text{and}\quad
 \delta(\Phi_{v,\beta})=
 \frac{\beta}{n}\max_{1\leq r\leq m}\frac{W_r^2}{r}.
 \label{eq:profile-parameters}
\end{equation}
For $S=\Phi_{v,\beta}-\Pi$, Cauchy--Schwarz in the ordinary
trace inner product gives
\[
 \|S(X)\|_2\leq\frac{\beta}{n}\Tr(D^2)\|X\|_2\quad\text{and}\quad \Tr(D^2)=2\sum_{j=1}^m w_j^2.
\]
Equality holds at $X=D$, proving the first formula.
For a rank-$r$ projection $P$ with $r\leq m$, its diagonal entries
in the eigenbasis of $D$ lie in $[0,1]$ and sum to $r$.
The largest trace pairing places all this weight on the $r$ largest
eigenvalues, giving $W_r$; the $r$ most negative eigenvalues give
$-W_r$. Thus the maximum of $|\Tr(DP)|$ is $W_r$, attained by the
leading positive spectral projection. The identity projection contributes
zero because $\Tr(D)=0$. Complementing a proper projection of rank above $n/2$
preserves this trace in absolute value and decreases its rank.
The maximum of $|\Tr(DP)|/\sqrt{\rank(P)}$ is therefore
$\max_{r\leq m}W_r/\sqrt r$. Factoring the two projection
traces gives the second formula.

Figure~\ref{fig:profiles} compares flat weights $w_j=1$,
harmonic weights $w_j=j^{-1/2}$
\citep[Section~2.2, Theorem~5]{BollobasNikiforov2004},
and square-root increments $w_j=\sqrt j-\sqrt{j-1}$.
For $n=2m$ and $\beta=1$, the harmonic stochastic matrix is precisely
their block matrix divided by $n$, up to coordinate ordering.
The evaluated grid contains 74 even
and odd dimensions up to $4097$, with
$\beta\in\{0,0.1,0.5,1\}$. \supporting{Exhaustive subset calculations for
$n=2,\ldots,11$ agree with \eqref{eq:profile-parameters}
to within $6\times10^{-17}$ over 120 cases.
The increment profile attains $\lambda/\delta=K_n$ in both
parities, as proved in Theorem~\ref{thm:channel-envelope}.}
For positive $\beta$, the ratio is independent of $\beta$;
at $\beta=0$, the channel is exactly depolarizing and both
parameters vanish.

\begin{figure}[t]
 \centering
 \includegraphics[width=\textwidth]{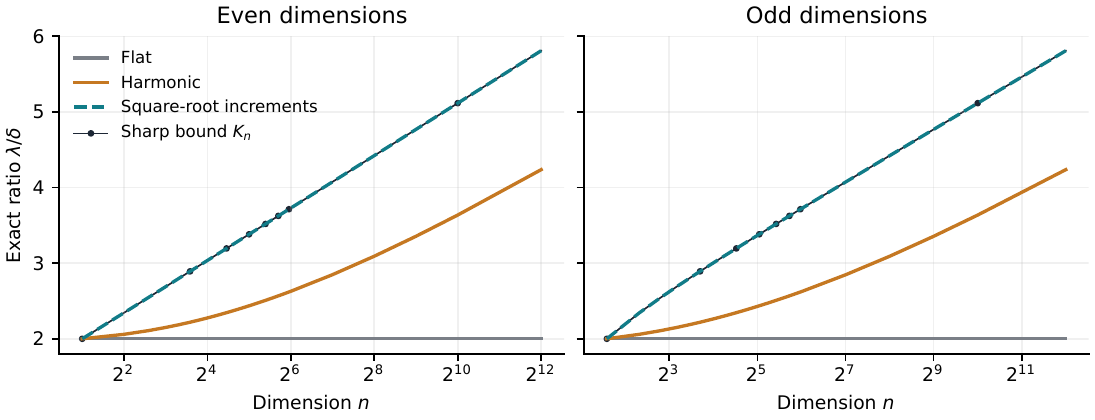}
 \caption{Exact spectral-to-projection ratios for three signed rank-one
 channel families at $\beta=1$. \supporting{Square-root increments attain the centered
 constant $K_n$ in both parities.} Flat spectra give ratio $2$, and
 harmonic spectra, following \citet{BollobasNikiforov2004}, exhibit
 the logarithmic separation. The curves
 evaluate explicit formulas; exhaustive small-dimensional subset
 calculations independently check their implementation.}
 \label{fig:profiles}
\end{figure}

\subsection{Qubit calibration and independent measurement axes}
\label{sec:qubit-calibration}

Let $\sigma_x,\sigma_y,\sigma_z$ be the Pauli matrices,
\[
 \sigma_x=\begin{pmatrix}0&1\\1&0\end{pmatrix},\qquad
 \sigma_y=\begin{pmatrix}0&-i\\i&0\end{pmatrix},\qquad
 \sigma_z=\begin{pmatrix}1&0\\0&-1\end{pmatrix}.
\]
Write $\sigma=(\sigma_x,\sigma_y,\sigma_z)$. For a Bloch vector
$r=(r_x,r_y,r_z)^\top\in\mathbb R^3$ with $\|r\|_{\ell_2}\leq1$,
set $r\cdot\sigma=r_x\sigma_x+r_y\sigma_y+r_z\sigma_z$.
A unital qubit channel has a real Bloch matrix
$M\in\mathbb R^{3\times3}$ defined by
$\Phi((I+r\cdot\sigma)/2)=(I+(Mr)\cdot\sigma)/2$.
Preparing a pure state, for which $\|r\|_{\ell_2}=1$, and measuring
along a unit vector $s\in\mathbb R^3$ gives Born bias
$|s^\top Mr|/2$ relative to complete depolarization.
The optimal pair consists of right and left singular axes, and
its value is $\|M\|_{\mathrm{op}}/2=\delta(\Phi)$ by
Corollary~\ref{cor:qubit}. Here $\|\cdot\|_{\mathrm{op}}$
is the Euclidean operator norm. The fixed computational-axis
test has value $|M_{zz}|/2$, where $M_{zz}$ is the Bloch-matrix entry
for the $z$ coordinate. The best test constrained
to use the same input and output axis has value
\[
 \frac12\left\|\frac{M+M^\top}{2}\right\|_{\mathrm{op}}.
\]
Indeed, $r^\top(M-M^\top)r=0$ for every real $r$, so a same-axis
pairing depends only on the symmetric part of $M$. Its maximum absolute
Rayleigh quotient over unit vectors is its largest absolute eigenvalue,
which equals the displayed operator norm.

Table~\ref{tab:qubit-benchmarks} compares these calibration tests
for Pauli channels and their unitary pre- and postprocessings.
We list Pauli probabilities in the order $(I,\sigma_x,\sigma_y,\sigma_z)$.
For nonnegative $p_0,p_x,p_y,p_z$ summing to one, the tuple
$(p_0,p_x,p_y,p_z)$ specifies the channel
$X\mapsto p_0X+p_x\sigma_xX\sigma_x+p_y\sigma_yX\sigma_y+
p_z\sigma_zX\sigma_z$.
The isotropic Pauli probabilities
are $(0.7,0.1,0.1,0.1)$. The anisotropic probabilities are
$(0.55,0.25,0.15,0.05)$, giving Bloch diagonal $(0.6,0.4,0.2)$.
The two individual output rotations have angle
$\pi/2$. The row with two rotations composes an input rotation
by $0.37$ radians about $x$ and an output rotation by $1.1$
radians about $y$ with the same anisotropic channel.

The last row gives a sharp failure of the same-axis restriction:
\[
 M=0.4
 \begin{pmatrix}
 0&-1&0\\
 1&0&0\\
 0&0&0
 \end{pmatrix}.
\]
This channel is obtained from Pauli probabilities
$(0.45,0.25,0.25,0.05)$ followed by the output unitary
$\exp(-i\pi\sigma_z/4)$, so it is completely positive, trace
preserving, and unital. \supporting{Since $r^\top Mr=0$ for every real $r$,
every same-axis test vanishes. Input axis $x$ and output axis
$y$ attain the nonzero optimum $0.2$.}
\supporting{Direct Kraus evaluations agree with the singular-axis biases
within $2\times10^{-16}$ in all six cases; direct normalized
Choi-matrix calculations also confirm positivity.}

\begin{table}[t]
 \centering
 \caption{Born biases of optimal and restricted qubit tests.
 \supporting{Independent input and output axes attain $\delta=\lambda/2$;
 output rotations preserve this value while changing which axes
 reveal it.} The two-rotation row is rounded to six decimals.
 Displayed zeros vanish analytically.}
 \label{tab:qubit-benchmarks}
 \begin{tabular}{lrrrr}
 \hline
 Channel & $\lambda$ & Two axes & Fixed $z$ & Same axis\\
 \hline
 Isotropic & $0.6$ & \supporting{$0.3$} & $0.3$ & $0.3$\\
 Anisotropic, aligned & $0.6$ & \supporting{$0.3$} & $0.1$ & $0.3$\\
 Anisotropic, output $z$ rotation & $0.6$ & \supporting{$0.3$} & $0.1$ & $0.1$\\
 Anisotropic, output $y$ rotation & $0.6$ & \supporting{$0.3$} & $0$ & $0.2$\\
 Anisotropic, two rotations & $0.6$ & \supporting{$0.3$} & $0.042290$ & $0.218493$\\
 Antisymmetric & $0.4$ & \supporting{$0.2$} & $0$ & $0$\\
 \hline
 \end{tabular}
\end{table}

\subsection{Balanced effects for parallel channel uses}
\label{sec:balanced-effects}

We evaluated the threshold construction of
Theorem~\ref{thm:regularized} for the increment extremizer in
dimension $64$ and a three-level staircase spectrum in dimension $32$, both
at $\beta=1$. For the latter, choose
\[
 w_1=1,\qquad w_2=w_3=w_4=\frac13,\qquad
 w_5=\cdots=w_{16}=\frac16.
\]
The partial sums $W_r$ lie on chords of $\sqrt r$ joining
$r=1,4,16$, so $W_r\leq\sqrt r$, with equality at those ranks.
Equation~\eqref{eq:profile-parameters} therefore gives
\[
 \lambda=\frac5{48},\qquad \delta_1=\frac1{32},\qquad
 \frac{\arcsin(5/48)}{\pi}=0.033217537757\ldots>\delta_1.
\]
This simple spectrum permits integer score accumulation after multiplication
by six. Both models have even dimension. A pair of spectral scores
with the exact classical joint law can be sampled as follows.
Choose independent uniform indices $a,b\in\{1,\ldots,m\}$ and
an independent fair sign $s\in\{-1,1\}$. Conditional on $(a,b)$,
choose $e\in\{-1,1\}$ independently of $s$ with
\[
 \Pr(e=1\mid a,b)=\frac{1+\beta w_aw_b}{2},\qquad
 X=sw_a,\quad Y=sew_b.
\]
Each labeled score pair then has probability
$(1+\beta v_iv_j)/n^2=A_{ij}/n$ as required. Independent
repetition generates parallel uses without forming a tensor
transition matrix. Odd dimensions enter only the deterministic
profile study above.

To obtain exactly balanced effects at every copy count, fix
$\kappa=10^{-12}$. For a real argument $u$, define the odd function
\[
 g_\kappa(u)=
 \begin{cases}
 -1,&u<-\kappa,\\
 0,&|u|\leq\kappa,\\
 1,&u>\kappa.
 \end{cases}
\]
For an integer $k\geq1$, let $(X_j,Y_j)$, $1\leq j\leq k$,
be independent copies of the sampled score pair. On a product input
label, the diagonal effect $E_k$ has value
$(1+g_\kappa(\sum_{j=1}^k X_j))/2$; define the output effect
$F_k$ using $\sum_{j=1}^kY_j$ in the same formula.
Thus acceptance in the fixed symmetric band
has probability $1/2$. This defines effects, including on ties.
By sign symmetry their normalized traces satisfy
$\tau_{n^k}(E_k)=\tau_{n^k}(F_k)=1/2$, where
$\tau_{n^k}=\Tr/n^k$. If
\[
 Z_k=g_\kappa\!\left(\sum_{j=1}^kX_j\right)
       g_\kappa\!\left(\sum_{j=1}^kY_j\right),
\]
then their centered correlation is $\mathbb EZ_k/4$ and their
trace-normalized witness value is
\begin{equation}
 \nu_k=
 \frac{|\tau_{n^k}[F_k(\Phi^{\otimes k}-\Pi_{n^k})(E_k)]|}
 {\sqrt{\tau_{n^k}(E_k)\tau_{n^k}(F_k)}}
 =\frac{|\mathbb EZ_k|}{2}\leq\delta(\Phi^{\otimes k}).
 \label{eq:empirical-effect}
\end{equation}
The inequality follows from Proposition~\ref{prop:effects}.
This is also the Born bias of the separable input
$E_k/\Tr(E_k)$ and measurement $F_k$, which uses local basis
measurements and classical thresholding. The fixed band has
width $\kappa/\sqrt k$ on the central-limit scale, so the
Gaussian threshold limit is unchanged.

We used $N=2\,000\,000$ independent trajectories for the increment model
and $N=8\,000\,000$ for the three-level model, with the counts
$k\in\{1,4,16,64,128\}$ fixed before sampling. Trajectory prefixes
are reused across counts. The two studies together use
$(2\,000\,000+8\,000\,000)\times128=1\,280\,000\,000$
sampled score pairs. The three-level study accumulates scores in integer
units of $1/6$, so ties are detected exactly. For trajectory
$\ell\in\{1,\ldots,N\}$, let $Z_k^{(\ell)}$ denote its threshold-product
value and define the sample mean
$\overline Z_k=N^{-1}\sum_{\ell=1}^N Z_k^{(\ell)}$.
The estimate is $\widehat\nu_k=|\overline Z_k|/2$. Since $Z_k/2$ lies in
$[-1/2,1/2]$, set $M=10$ for the total number of settings.
Hoeffding's inequality~\citep{Hoeffding1963} gives, for each setting,
a confidence radius $r_N=\sqrt{\log(2M/0.01)/(2N)}$ with
\[
 \Pr\!\left(\left|\frac{\overline Z_k-\mathbb EZ_k}{2}\right|>r_N\right)
 \leq2e^{-2Nr_N^2}=\frac{0.01}{M}.
\]
The reverse triangle inequality gives
$|\widehat\nu_k-\nu_k|\leq|\overline Z_k-\mathbb EZ_k|/2$.
A union bound over the $M$ settings therefore gives simultaneous
$99\%$ coverage. Intersecting with the known range $[0,1/2]$
yields the intervals
\begin{equation}
 [\max(0,\widehat\nu_k-r_N),\ \min(1/2,\widehat\nu_k+r_N)],
 \qquad
 r_N=\sqrt{\frac{\log(2M/0.01)}{2N}}.
 \label{eq:empirical-confidence}
\end{equation}
The radii are approximately $0.00137849$ for the increment model and
$0.00068924$ for the three-level model.
Independence between counts is unnecessary for this union bound.
These intervals quantify sampling uncertainty for the specified
effect, without bounding floating-point systematic error. Their
lower endpoints are lower confidence bounds on the optimized
discrepancy; their upper endpoints are not upper bounds on it.
Table~\ref{tab:threshold-all-settings} reports the signed means,
absolute witness estimates, band counts, and intervals for all ten
settings. The independent units for each mean are complete trajectories.

\begin{table}[t]
\centering\small
\caption{All ten displayed sampled settings. Each row uses $N$ independent
trajectories; the same trajectories supply the five prefixes within each model.
$T_X,T_Y$ count input/output sums in the tie band $[-10^{-12},10^{-12}]$;
for the staircase model these are exact zero sums.
Intervals have simultaneous $99\%$ coverage for the specified effect values
and are rounded outward. Their lower endpoints also bound the optimized
discrepancy from below; their upper endpoints do not bound that optimum.}
\label{tab:threshold-all-settings}
\resizebox{\textwidth}{!}{%
\begin{tabular}{lr r r r r r l}
\hline
Model & $k$ & $N$ & $\overline Z_k$ & $|\overline Z_k|/2$ & $T_X$ & $T_Y$ & Interval\\
\hline
Increment, $n=64$ & 1 & 2000000 & 0.0311110 & 0.01555550 & 0 & 0 & $[0.014177,0.016934]$\\
Increment, $n=64$ & 4 & 2000000 & 0.0320070 & 0.01600350 & 1418 & 1404 & $[0.014625,0.017382]$\\
Increment, $n=64$ & 16 & 2000000 & 0.0358900 & \supporting{0.01794500} & 0 & 0 & \supporting{$[0.016566,0.019324]$}\\
Increment, $n=64$ & 64 & 2000000 & 0.0375860 & \supporting{0.01879300} & 0 & 0 & \supporting{$[0.017414,0.020172]$}\\
Increment, $n=64$ & 128 & 2000000 & 0.0362800 & \supporting{0.01814000} & 0 & 0 & \supporting{$[0.016761,0.019519]$}\\
Staircase, $n=32$ & 1 & 8000000 & 0.0626990 & 0.03134950 & 0 & 0 & $[0.030660,0.032039]$\\
Staircase, $n=32$ & 4 & 8000000 & 0.0598051 & \limitation{0.02990256} & 1219827 & 1218801 & \limitation{$[0.029213,0.030592]$}\\
Staircase, $n=32$ & 16 & 8000000 & 0.0647519 & \supporting{0.03237594} & 425280 & 423739 & \supporting{$[0.031686,0.033066]$}\\
Staircase, $n=32$ & 64 & 8000000 & 0.0661863 & \supporting{0.03309313} & 207270 & 206971 & \supporting{$[0.032403,0.033783]$}\\
Staircase, $n=32$ & 128 & 8000000 & 0.0659712 & \supporting{0.03298562} & 146840 & 145802 & \supporting{$[0.032296,0.033675]$}\\
\hline
\end{tabular}}
\end{table}

For the dimension-$64$ channel,
$\lambda=0.0582528652\ldots$ and the exact one-use optimum is
$\delta_1=1/64=0.015625$. \supporting{The lower confidence bounds exceed
$\delta_1$ at $k=16,64,128$. At $k=64$, the estimate is
$0.018793$, with interval approximately
$[0.0174145,0.0201715]$, compared with the Gaussian threshold
reference $\arcsin(\lambda)/\pi=0.01855297\ldots$.}
\supporting{For the three-level channel, the lower confidence bounds also exceed
its exact one-use optimum $\delta_1=0.03125$ at $k=16,64,128$.
At $k=64$, the estimate is $0.033093125$, approximately $5.9\%$
above $\delta_1$, with interval $[0.0324038,0.0337824]$ rounded outward.
Thus finite-copy improvement occurs for both the extremal increment
spectrum and this simpler three-level spectrum.}
\supporting{Independent rational convolution at $k=1,4$ agrees with the corresponding
sampling intervals.} Table~\ref{tab:threshold-all-settings} reports all five
copy counts for each model; the exact one-use baselines and Gaussian
limits are given above.

For an equal-channel-use comparison in the dimension-$64$
model, let $P_+$ project onto its 32 positive coordinates.
The pair $P_+,P_+$ attains the one-use $\delta$ optimum.
Input $P_+/32$ and measurement $P_+$ have acceptance
probabilities $p_0=1/2$ under $\Pi$ and $p_1=33/64$ under
$\Phi$, since $\Tr(DP_+)=\sqrt{32}$.
Repeat this product-input test 64 times and accept at least
33 successes. For $b\in\{0,1\}$, let $q_b$ be the acceptance
probability when each trial has success probability $p_b$. Then
\begin{equation}
 q_b=\sum_{j=33}^{64}\binom{64}{j}p_b^j(1-p_b)^{64-j}\quad\text{and}\quad
 q_1-q_0=0.09988200755\ldots.
 \label{eq:repeated-bias}
\end{equation}
The selected balanced effect's Born bias equals its discrepancy score and
is at most $\lambda/2$. Indeed, for effects $E,F\in M_{n^k}(\mathbb C)$ with
$\tau_{n^k}(E)=\tau_{n^k}(F)=1/2$, put $E_0=E-I/2$ and $F_0=F-I/2$.
Since $E^2\leq E$ and $F^2\leq F$,
\[
 \|E_0\|_2^2=\tau_{n^k}(E^2)-\tfrac14\leq\tfrac14\quad\text{and}\quad \|F_0\|_2^2\leq\tfrac14.
\]
For $S_k=\Phi^{\otimes k}-\Pi_{n^k}$, centering and
Theorem~\ref{thm:regularized} give
\[
 |\tau_{n^k}[F S_k(E)]|
 =|\langle F_0,S_k(E_0)\rangle|
 \leq\lambda\|F_0\|_2\|E_0\|_2\leq\lambda/4.
\]
Both the discrepancy score and the Born bias are twice this
absolute correlation, so
\begin{equation}
 \nu_k\leq\lambda/2.
 \label{eq:balanced-effect-upper}
\end{equation}
The repeated-majority test uses the input projection $P_+^{\otimes64}$,
whose normalized rank is $2^{-64}$, and its majority output
projector has normalized rank $q_0=0.4503266231\ldots$.
This test's rank-normalized discrepancy score is consequently
\[
 (q_1-q_0)\sqrt{2^{-64}/q_0}=3.46548\ldots\times10^{-11}.
\]
The balanced effects have both normalized traces $1/2$.
The comparison distinguishes the two objectives while retaining the valid
one-use discrepancy baselines stated above.

\FloatBarrier
\appendix
\section{Sharp comparisons for unrestricted Hermiticity-preserving maps}
\label{sec:unrestricted}

Without the centering identities, the exact constant is
\[
 C_n=c_{\lfloor n/2\rfloor}+c_{\lceil n/2\rceil},
\]
with $c_m$ as in \eqref{eq:dimension-constants}. It agrees with
$K_n$ for even $n$. The proof below preserves the positive and
negative spectral layers separately and gives an explicit
$n^2$-pair witness list for general Hermiticity-preserving maps.

\begin{lemma}[Cost of a signed spectral decomposition]
\label{lem:spectral-cost}
Every Hermitian matrix $H$ has a decomposition
$H=\sum_j\alpha_jP_j$ into at most $n$ nonzero projections and
real coefficients such that
\begin{equation}
 L(H):=\sum_j|\alpha_j|\sqrt{\rank(P_j)}
 \leq\sqrt{C_n}\,\|H\|_{\mathrm F}.
 \label{eq:spectral-cost}
\end{equation}
The projections can be chosen as the positive and negative
spectral prefixes of $H$.
\end{lemma}

\begin{proof}
Let $p$ and $q$ be the numbers of positive and negative eigenvalues
of $H$, respectively. Let $a_1\geq\cdots\geq a_p>0$ be its positive
eigenvalues, and let $b_1\geq\cdots\geq b_q>0$ be the magnitudes of its negative
eigenvalues. Set $a_{p+1}=b_{q+1}=0$. If $P_k^+$ projects onto the
first $k$ positive eigenvectors and $P_\ell^-$ onto the first
$\ell$ negative eigenvectors in these orders, then
\[
 H=\sum_{k=1}^p(a_k-a_{k+1})P_k^+
   -\sum_{\ell=1}^q(b_\ell-b_{\ell+1})P_\ell^-.
\]
Zero coefficients may be omitted; the remaining prefixes are
threshold spectral projections even when eigenvalues repeat.
Summation by parts and Cauchy--Schwarz give
\[
L(H)
 =\sum_{k=1}^p a_kd_k+\sum_{\ell=1}^q b_\ell d_\ell
 \leq\sqrt{c_p+c_q}\,
       \left(\sum_{k=1}^p a_k^2+\sum_{\ell=1}^q b_\ell^2\right)^{1/2}.
\]
The increments $d_k^2$ are positive and decreasing. For fixed
$p+q$, the sum $c_p+c_q$ is maximized when the two indices differ
by at most one: moving one unit from $q$ to $p$ when $q\geq p+2$
changes the sum by $d_{p+1}^2-d_q^2\geq0$.
Since $p+q\leq n$, positivity of the increments then gives
$c_p+c_q\leq C_n$. Empty eigenvalue blocks are covered by $c_0=0$.
The zero matrix has the empty decomposition. This proves
\eqref{eq:spectral-cost}.
\end{proof}

\begin{theorem}[Sharp dimension comparison and finite witnesses]
\label{thm:dimension}
For every Hermiticity-preserving complex-linear map $S:\M\to\M$,
\begin{equation}
 \|S\|_{\proj}\geq\frac{\|S\|_{2\to2}}{C_n}.
 \label{eq:dimension-bound}
\end{equation}
The constant $C_n$ is optimal in this class for every $n\geq2$,
and $C_n=\frac12\log n+O(1)$ as $n\to\infty$.

More precisely, for any nonzero Hermitian $H$ with $S(H)\ne0$,
one of at most $n^2$ pairs of positive or negative spectral-prefix
projections $P$ of $H$ and $Q$ of $S(H)$ satisfies
\begin{equation}
 \frac{|\Tr[Q S(P)]|}{\sqrt{\rank(P)\rank(Q)}}
 \geq\frac{\|S(H)\|_2}{C_n\|H\|_2}.
 \label{eq:prefix-witness}
\end{equation}
\end{theorem}

\begin{proof}
Put $Y=S(H)$ and expand $H$ and $Y$ using
Lemma~\ref{lem:spectral-cost}. Let $M$ be the largest normalized
absolute correlation over the pairs of projections in these
two decompositions. Since $Y$ is Hermitian,
\[
 \|Y\|_{\mathrm F}^2
   =|\Tr[Y S(H)]|
   \leq M L(H)L(Y)
   \leq M C_n\|H\|_{\mathrm F}\|Y\|_{\mathrm F}.
\]
Cancel $\|Y\|_{\mathrm F}>0$ to obtain
\eqref{eq:prefix-witness}. Each decomposition contains at most
$n$ projections, so at most $n^2$ pairs are tested.
For a nonzero map, choose a Hermitian norm-attaining input by
Lemma~\ref{lem:hermitian-attainment}; this proves
\eqref{eq:dimension-bound}. The zero map is immediate.

To prove optimality, set $p=\lfloor n/2\rfloor$,
$q=\lceil n/2\rceil$, and
\[
 D=\diag(d_1,\ldots,d_p,-d_1,\ldots,-d_q)\quad\text{and}\quad
 S_D(X)=\Tr(DX)D.
\]
The rank-one map $S_D$ has
$\|S_D\|_{2\to2}=\Tr(D^2)=c_p+c_q=C_n$.
For any rank-$r$ projection $P$, its diagonal entries lie in
$[0,1]$ and sum to $r$. Its positive contribution to $\Tr(DP)$
is at most
$\sum_{k=1}^{\min(r,p)}d_k=\sqrt{\min(r,p)}$,
and its negative contribution in absolute value is at most
$\sqrt{\min(r,q)}$. Thus
\begin{equation}
 |\Tr(DP)|\leq\sqrt r.
 \label{eq:extremal-projection}
\end{equation}
The first coordinate projection attains equality at $r=1$.
Factoring the rank-one correlation therefore gives
\[
 \|S_D\|_{\proj}
 =\left(\max_{P\ne0}
       \frac{|\Tr(DP)|}{\sqrt{\rank(P)}}\right)^2=1.
\]
This proves that the constant cannot be decreased.

Finally,
$d_k=1/(\sqrt{k}+\sqrt{k-1})$, so for $k\geq2$,
\[
 \frac1{4k}\leq d_k^2\leq\frac1{4(k-1)}.
\]
Summation yields $c_m=\frac14\log m+O(1)$, and hence the
stated asymptotic for $C_n$.
\end{proof}

\begin{remark}[What the finite search guarantees]
Fix $0<\alpha\leq1$. If a nonzero Hermitian $H$ with $S(H)\ne0$
satisfies $\|S(H)\|_2\geq\alpha\|S\|_{2\to2}\|H\|_2$, then \eqref{eq:prefix-witness} gives a witness
of size at least $\alpha\|S\|_{2\to2}/C_n$.
One computes the two spectral decompositions and evaluates the
candidate correlations. 
\end{remark}

\section{Height bounds by diagonal compression}
\label{sec:proof}

Lev's matrix inequality applies beyond channels, including maps
without positivity or centering. We give the general two-basis
compression and its height-dependent projection estimate here;
the primary channel proof in Section~\ref{sec:mixing} is independent
of this appendix. For a nonzero complex $n\times n$ matrix $B$, regarded here as a map on
vectors, define
\[
 h_{\mathrm{cl}}(B)=
 \frac{\sqrt{\|B\|_{\ell_1\to\ell_1}
                   \|B\|_{\ell_\infty\to\ell_\infty}}}
      {\|B\|_{\ell_2\to\ell_2}}.
\]
The vector norms here are the ordinary $\ell_p$ norms defined in
Section~\ref{sec:norms}. Scaling both the input and output norms by
$n^{-1/p}$ leaves each induced $p\to p$ norm unchanged.

\begin{theorem}[Lev {\citep[Theorem 2]{lev2015discrete}}]
\label{thm:lev}
For every nonzero complex $n\times n$ matrix $B$,
\begin{equation}
 \max_{\substack{a,b\in\{0,1\}^n\\a,b\ne0}}
 \frac{|b^\dagger Ba|}
      {\sqrt{\left(\sum_i a_i\right)\left(\sum_i b_i\right)}}
 \geq
 \frac{\|B\|_{\ell_2\to\ell_2}}
 {32\sqrt2\,(\log h_{\mathrm{cl}}(B)+4)}.
 \label{eq:lev}
\end{equation}
\end{theorem}

The denominator in \eqref{eq:lev} is the product of the unnormalized
Euclidean norms of the binary vectors, as in the classical estimate.

\subsection{A compression preserving the spectral norm}

\begin{lemma}[Diagonal compression]
\label{lem:compression}
Let $S:\M\to\M$ be a nonzero Hermiticity-preserving complex-linear
map. There are unitary matrices $U,V\in M_n(\mathbb C)$ such that
the real matrix $B\in\mathbb R^{n\times n}$, defined for
$x\in\mathbb C^n$ by
\begin{equation}
 Bx=\diag\!\left(V^\dagger
          S\bigl(U\diag(x)U^\dagger\bigr)V\right)
 \label{eq:compression}
\end{equation}
satisfies
\begin{equation}
\|B\|_{\ell_p\to\ell_p}\leq\|S\|_{p\to p}
       \quad (p=1,2,\infty) \quad\text{and}\quad \|B\|_{\ell_2\to\ell_2}=\|S\|_{2\to2}.
 \label{eq:compression-bounds}
\end{equation}
In \eqref{eq:compression}, $\diag$ denotes diagonal insertion on
vectors and diagonal extraction on matrices.
\end{lemma}

\begin{proof}
By Lemma~\ref{lem:hermitian-attainment}, choose a nonzero Hermitian
$H$ satisfying $\|S(H)\|_2=\|S\|_{2\to2}\|H\|_2$.
Choose unitaries $U,V$ and real vectors $x,y$ such that
\[
 H=U\diag(x)U^\dagger\quad\text{and}\quad
 S(H)=V\diag(y)V^\dagger.
\]
The matrix $B$ in \eqref{eq:compression} is real, because each column
is the diagonal of a Hermitian matrix.

Diagonal insertion satisfies
$\|\diag(x)\|_p=n^{-1/p}\|x\|_{\ell_p}$, with $1/\infty=0$.
Diagonal extraction is contractive in Schatten $p$ norm:
the diagonal matrix obtained from a matrix $Z$ by deleting its
off-diagonal entries is
\[
 2^{-n}\sum_{\eta\in\{-1,1\}^n}D_\eta ZD_\eta^\dagger\quad\text{and}\quad D_\eta=\diag(\eta).
\]
Unitary invariance and the triangle inequality prove contraction
for every Schatten $p$ norm, including $p=\infty$.
Consequently, for every $z\in\mathbb C^n$ and $p\in\{1,2,\infty\}$,
\[
 \|Bz\|_{\ell_p}
 \leq n^{1/p}\|S(U\diag(z)U^\dagger)\|_p
 \leq\|S\|_{p\to p}\|z\|_{\ell_p}.
\]
Here diagonal extraction contributes $n^{1/p}$ and diagonal
insertion contributes $n^{-1/p}$, so the dimension factors cancel.
Taking the supremum over nonzero $z$ proves
\eqref{eq:compression-bounds}. For the chosen $x$, $Bx=y$, and therefore
\[
 \frac{\|Bx\|_{\ell_2}}{\|x\|_{\ell_2}}
 =\frac{\|S(H)\|_2}{\|H\|_2}
 =\|S\|_{2\to2}.
\]
Together with the upper bound, this proves
\eqref{eq:compression-bounds}.
\end{proof}

\subsection{Projection witnesses for Hermiticity-preserving maps}

For a nonzero superoperator $S$, define its height by
\begin{equation}
 h(S)=\frac{\sqrt{\|S\|_{1\to1}\|S\|_{\infty\to\infty}}}
                 {\|S\|_{2\to2}}.
 \label{eq:height}
\end{equation}
Schatten interpolation gives $h(S)\geq1$
\citep{gillespie1991noncommutative}. For Hermiticity-preserving maps,
it also follows from Lemma~\ref{lem:compression} and the matrix bound
\[
 \|B\|_{\ell_2\to\ell_2}^2\leq
 \|B\|_{\ell_1\to\ell_1}\|B\|_{\ell_\infty\to\ell_\infty}.
\]

\begin{theorem}[Projection estimate]
\label{thm:projection}
Every nonzero Hermiticity-preserving complex-linear map
$S:\M\to\M$ satisfies
\begin{equation}
 \frac{\|S\|_{2\to2}}
 {32\sqrt2\,(\log h(S)+4)}
 \leq\|S\|_{\proj}\leq\|S\|_{2\to2}.
 \label{eq:projection-theorem}
\end{equation}
The constant is independent of the dimension and the map; the
witnessing projections may depend on $S$.
\end{theorem}

\begin{proof}
For nonzero projections $P,Q$, Cauchy--Schwarz gives
\[
 |\ip{Q}{S(P)}|
 \leq\|Q\|_2\|S(P)\|_2
 \leq\|S\|_{2\to2}\|Q\|_2\|P\|_2.
\]
Dividing by $\|Q\|_2\|P\|_2$ and maximizing proves the upper bound.
For the lower bound, choose $U,V,B$ as in
Lemma~\ref{lem:compression}. Then
\[
 h_{\mathrm{cl}}(B)\leq h(S)\quad\text{and}\quad
 \|B\|_{\ell_2\to\ell_2}=\|S\|_{2\to2}.
\]
Apply Theorem~\ref{thm:lev} to obtain nonzero binary vectors $a,b$.
The matrices
\[
 P=U\diag(a)U^\dagger\quad\text{and}\quad
 Q=V\diag(b)V^\dagger
\]
are nonzero projections, and
\[
 \Tr[Q S(P)]=b^\dagger Ba,\qquad
 \rank(P)=\sum_i a_i,\qquad \rank(Q)=\sum_i b_i.
\]
Substituting these identities in \eqref{eq:lev} proves the lower
bound in \eqref{eq:projection-theorem}.
\end{proof}

\begin{corollary}[Channel specialization of Lev's estimate]
\label{cor:lev-channel}
For a unital quantum channel with $\lambda=\lambda(\Phi)>0$,
\begin{equation}
 \delta(\Phi)\geq
 \frac{\lambda}{32\sqrt2\,(\log(2/\lambda)+4)}.
 \label{eq:lev-channel}
\end{equation}
\end{corollary}

\begin{proof}
Set $S=\Phi-\Pi$, a map on the full matrix algebra.
If $\lambda=\|S\|_{2\to2}>0$, Lemma~\ref{lem:channel-norms} gives
\[
 h(S)\leq\frac2\lambda.
\]
Theorem~\ref{thm:projection} now yields
\[
 \delta(\Phi)=\|S\|_{\proj}
 \geq\frac{\lambda}{32\sqrt2\,(\log(2/\lambda)+4)}.
\]
The application is to $S$ on $\M$, whose domain contains the
projections, rather than to the restriction of $\Phi$ to the
traceless subspace.

Since $S(I)=S^*(I)=0$, replacing $P$ by $I-P$ or $Q$ by $I-Q$
preserves the absolute value of the correlation. A nonzero
correlation excludes $P=I$ and $Q=I$. Complementing a projection
whose rank exceeds $n/2$ decreases the denominator in
\eqref{eq:projection-norm}, so both witness ranks can be taken
at most $n/2$. The case $\lambda=0$ follows from $S=0$.
\end{proof}

\subsection{The classical normalization}

Let $A$ be the adjacency matrix of an undirected $d$-regular
graph on the vertex set $[n]$, divided by its degree $d\geq1$.
For subsets $S,T\subseteq[n]$, let $e(S,T)$ count ordered
adjacent pairs $(u,v)\in S\times T$. Thus edges with both
endpoints in $S\cap T$ are counted twice. Then
\[
 \mathbf1_T^\top(A-J_n/n)\mathbf1_S
   =\frac{e(S,T)}{d}-\frac{|S||T|}{n}.
\]
Set $\lambda=\|A-J_n/n\|_{\ell_2\to\ell_2}$.
If $\lambda>0$, Corollary~\ref{cor:lev-channel} gives
nonempty $S,T$ with
\begin{equation}
 \left|\frac{e(S,T)}d-\frac{|S||T|}n\right|
 \geq
 \frac{\lambda}{32\sqrt2\,(\log(2/\lambda)+4)}
 \sqrt{|S||T|},\qquad \lambda>0.
 \label{eq:classical-converse}
\end{equation}
There is no extra factor $d$ on the right when the edge count on
the left is already divided by $d$. In this self-adjoint classical
case, $\lambda$ is the largest absolute nontrivial eigenvalue of
the normalized adjacency matrix, including the possibility
$\lambda=1$.

\section{Projection witnesses for arbitrary superoperators}
\label{sec:general}

The channel application needs only Hermiticity-preserving maps.
The same method also gives a projection estimate for an arbitrary
complex-linear map, with a change in constants. In dimension two,
write $|+\rangle=(|1\rangle+|2\rangle)/\sqrt2$.
Hermiticity preservation cannot simply be omitted from the
centered $K_n$ bound: on $M_2(\mathbb C)$, the centered map
$S(X)=X_{12}E_{12}$ has $\|S\|_{2\to2}=1$ but
$\|S\|_{\proj}=1/4<1/K_2$. Indeed,
$|\Tr[QS(P)]|=|P_{12}Q_{21}|\leq1/4$ for rank-one
projections, with equality at $P=Q=|+\rangle\langle+|$;
pairs involving $I$ give zero.

\begin{theorem}
\label{thm:general}
Every nonzero complex-linear map $S:\M\to\M$ satisfies
\begin{equation}
 \frac{\|S\|_{2\to2}}
 {64\sqrt2\,(\log(2h(S))+4)}
 \leq\|S\|_{\proj}\leq\|S\|_{2\to2}.
 \label{eq:general}
\end{equation}
\end{theorem}

\begin{proof}
Define the complex-linear map
$S^\sharp(X)=S(X^\dagger)^\dagger$ and its Hermiticity-preserving
components
\[
 S_{\mathrm R}=\frac{S+S^\sharp}{2},\qquad
 S_{\mathrm I}=\frac{S-S^\sharp}{2i},\qquad
 S=S_{\mathrm R}+iS_{\mathrm I}.
\]
For Hermitian $H$, these components satisfy
$S_{\mathrm R}(H)=(S(H)+S(H)^\dagger)/2$ and
$S_{\mathrm I}(H)=(S(H)-S(H)^\dagger)/(2i)$, so both outputs are
Hermitian. Thus Theorem~\ref{thm:projection} applies to either nonzero
component.
Adjunction is an isometry for each Schatten norm, so
$\|S^\sharp\|_{p\to p}=\|S\|_{p\to p}$ and
\[
 \|S_{\mathrm R}\|_{p\to p},\ \|S_{\mathrm I}\|_{p\to p}
 \leq\|S\|_{p\to p}\quad(p=1,2,\infty).
\]
The triangle inequality implies that one component, denoted $R$,
satisfies $\|R\|_{2\to2}\geq\|S\|_{2\to2}/2$. In particular it is
nonzero. Combining the endpoint norm bounds with this lower bound gives
\[
 h(R)=\frac{\sqrt{\|R\|_{1\to1}\|R\|_{\infty\to\infty}}}
                  {\|R\|_{2\to2}}
 \leq\frac{\sqrt{\|S\|_{1\to1}\|S\|_{\infty\to\infty}}}
               {\|S\|_{2\to2}/2}=2h(S).
\]

For Hermitian projections $P,Q$, the two pairings
$\ip{Q}{S_{\mathrm R}(P)}$ and
$\ip{Q}{S_{\mathrm I}(P)}$ are respectively the real and imaginary
parts of $\ip{Q}{S(P)}$. Hence $\|S\|_{\proj}\geq\|R\|_{\proj}$.
Theorem~\ref{thm:projection} yields
\[
 \|S\|_{\proj}\geq
 \frac{\|R\|_{2\to2}}{32\sqrt2\,(\log h(R)+4)}
 \geq
 \frac{\|S\|_{2\to2}}{64\sqrt2\,(\log(2h(S))+4)}.
\]
The upper bound is again Cauchy--Schwarz.
\end{proof}

\begin{remark}
The map is assumed nonzero so that its height is defined.
For the zero map the projection norm is zero. Both the absolute value in
\eqref{eq:projection-norm} and the dependence of $P,Q$ on $S$ are
essential to the general statement.
\end{remark}

\begingroup
\renewcommand{\doi}[1]{\href{https://doi.org/#1}{doi:\nolinkurl{#1}}}
\bibliography{bib-ms}
\endgroup

\clearpage
\section*{Notation}
The table collects the main notation. Local uses are identified by their
section or construction; indices and auxiliary vectors are defined where
they occur. All logarithms are natural.

\begingroup
\small
\renewcommand{\arraystretch}{1.0}
\setlength{\LTpre}{6pt}
\setlength{\LTpost}{0pt}
\begin{longtable}{@{}>{\raggedright\arraybackslash}p{0.22\textwidth}
                    >{\raggedright\arraybackslash}p{0.59\textwidth}
                    >{\raggedright\arraybackslash}p{0.12\textwidth}@{}}
\caption{Notation and conventions.}\label{tab:notation}\\
\hline
\textbf{Symbol} & \textbf{Meaning} & \textbf{Location}\\
\hline
\endfirsthead
\multicolumn{3}{@{}l}{\tablename~\thetable\ (continued)}\\
\hline
\textbf{Symbol} & \textbf{Meaning} & \textbf{Location}\\
\hline
\endhead
\hline
\multicolumn{3}{r@{}}{Continued on the next page}\\
\endfoot
\hline
\endlastfoot

\multicolumn{3}{@{}l}{\textit{Spaces, traces, and norms}}\\[2pt]
$n,N,k$ & Single-use dimension $n\geq2$, a general ambient dimension $N$,
and the number $k\geq1$ of parallel uses. In the sampling study, $N$
counts trajectories; in spectral-prefix formulas, $k$ is a rank index. & Secs.~\ref{sec:norms}, \ref{sec:parallel}\\[1pt]
$M_n(\mathbb C),I_n$ & Complex $n\times n$ matrices and their identity;
$I$ denotes the identity in the ambient dimension. & Sec.~\ref{sec:norms}\\[1pt]
$[n],\mathbf1_E,J_n$ & $[n]=\{1,\ldots,n\}$; $\mathbf1_E$ is a subset's
indicator vector; $J_n=\mathbf1\mathbf1^\top$ is the all-ones matrix. & Sec.~\ref{sec:mixing}\\[1pt]
$e_i$ & Standard coordinate vector with its $i$th entry one and all
other entries zero. & Cor.~\ref{cor:cut-envelope-endpoint}\\[1pt]
$\Tr,\tau_N,\tau$ & Ordinary trace; normalized trace
$\tau_N(X)=\Tr(X)/N$; $\tau$ uses the ambient dimension.
For $k$ copies the trace is $\tau_{n^k}$. & Sec.~\ref{sec:preliminaries}\\[1pt]
$\Pi_N,\Pi$ & Completely depolarizing channel
$\Pi_N(X)=\tau_N(X)I_N$; $\Pi$ uses the ambient dimension. & Sec.~\ref{sec:preliminaries}\\[1pt]
$X^\dagger,X^\top,S^*$ & Matrix adjoint, matrix transpose in the fixed basis,
and the adjoint of a superoperator for the Hilbert--Schmidt inner product. & Sec.~\ref{sec:preliminaries}\\[1pt]
$\diag(x),\diag(X)$ & Diagonal matrix with entries $x_i$, and the vector
of diagonal entries of a matrix $X$, respectively. & Sec.~\ref{sec:norms}\\[1pt]
$\langle X,Y\rangle$,\newline $\|X\|_p$ & Normalized inner product
$\tau(X^\dagger Y)$ and Schatten norm $(\tau(|X|^p))^{1/p}$;
$\|X\|_\infty$ is the largest singular value. & Sec.~\ref{sec:norms}\\[1pt]
$\|X\|_{\mathrm F}$ & Ordinary Frobenius norm
$\sqrt{\Tr(X^\dagger X)}=\sqrt n\,\|X\|_2$ on $M_n(\mathbb C)$. & Sec.~\ref{sec:norms}\\[1pt]
$\|S\|_{p\to q}$ & Induced map norm
$\sup_{X\ne0}\|S(X)\|_q/\|X\|_p$ for normalized Schatten norms. & Sec.~\ref{sec:norms}\\[1pt]
$\|x\|_{\ell_p}$,\newline $\|B\|_{\ell_p\to\ell_q}$ & Ordinary vector norm
$(\sum_i|x_i|^p)^{1/p}$ for finite $p$, the maximum norm for $p=\infty$,
and the induced matrix norm;
$\|B\|_{\mathrm{op}}=\|B\|_{\ell_2\to\ell_2}$. & Sec.~\ref{sec:norms}\\[4pt]

\multicolumn{3}{@{}l}{\textit{Channels and discrepancy parameters}}\\[2pt]
$\Phi,S,I^\perp$ & Unital quantum channel; its centered map $S=\Phi-\Pi$;
the traceless subspace $I^\perp=\{X:\Tr(X)=0\}$.
General-map results state their hypotheses on $S$ explicitly. & Sec.~\ref{sec:preliminaries}\\[1pt]
$P,Q,r,s,p,q$ & Input/output projections, their ranks $r,s$, and normalized
ranks $p=\tau(P)=r/n$, $q=\tau(Q)=s/n$. & Eq.~\eqref{eq:rank-formula}\\[1pt]
$E,F,\rho_P$ & Effects $0\leq E,F\leq I$ and the state
$\rho_P=P/\rank(P)$ for nonzero $P$. & Sec.~\ref{sec:norms}\\[1pt]
$J(\Phi),\omega_\Phi$ & Unnormalized Choi matrix and normalized Choi state
$\omega_\Phi=J(\Phi)/n$. & Eq.~\eqref{eq:choi-state}\\[1pt]
$K_a,d$ & Kraus operators and the number of displayed unitary terms
in a mixed-unitary representation, respectively. & Sec.~\ref{sec:preliminaries}\\[1pt]
$\lambda(\Phi)$ & Spectral parameter $\|\Phi-\Pi\|_{2\to2}$;
abbreviated to $\lambda$ for a fixed channel. & Eq.~\eqref{eq:intro-lambda}\\[1pt]
$\|S\|_{\proj}$,\newline $\delta(\Phi)$ & Projection-test norm
$\max_{P,Q\ne0}|\Tr[QS(P)]|/\sqrt{\rank(P)\rank(Q)}$
over nonzero orthogonal projections $P,Q$;
$\delta(\Phi)=\|\Phi-\Pi\|_{\proj}$. & Eq.~\eqref{eq:projection-norm}\\[1pt]
$\|S\|_{\cut}$,\newline $\epsilon(\Phi)$ & Cut norm
$\max_{P,Q\text{ projections}}|\tau[QS(P)]|$;
$\epsilon(\Phi)=\|\Phi-\Pi\|_{\cut}$. & Eq.~\eqref{eq:intro-epsilon}\\[1pt]
$\zeta(\Phi)$ & Maximum of $|\tau[Q(\Phi-\Pi)(P)]|/
\sqrt{p(1-p)q(1-q)}$, over projections with $0<p,q<1$. & Eq.~\eqref{eq:zeta}\\[1pt]
$\rho(A),\tau'(A)$,\newline $\mu(\omega_\Phi)$ & Classical maximal correlation,
Peyre's variance-normalized event coefficient, and quantum maximal
correlation. The selected compression satisfies
$\rho(A)=\mu(\omega_\Phi)=\lambda(\Phi)$. & Sec.~\ref{sec:mixing}; Eq.~\eqref{eq:choi-maximal}\\[1pt]
$f,f^{-1}$ & $f(t)=t(1-\log t)$ for $0<t\leq1$, $f(0)=0$,
and its inverse on $[0,1]$. & Sec.~\ref{sec:mixing}\\[1pt]
$\mathcal Q_n,\mathcal Q_n(\ell)$ & Set of unital quantum channels on
$M_n(\mathbb C)$ and its subset with $\lambda(\Phi)=\ell$, respectively.
& Sec.~\ref{sec:projection-envelopes}\\[1pt]
$d_n(\ell),e_n(\ell)$ & Minima of $\delta$ and $\epsilon$, respectively,
over $\mathcal Q_n(\ell)$. & Eq.~\eqref{eq:envelope-definitions}\\[4pt]

\multicolumn{3}{@{}l}{\textit{Sharp constants and finite witnesses}}\\[2pt]
$d_j,c_m$ & Square-root increments $d_j=\sqrt j-\sqrt{j-1}$ and
$c_m=\sum_{j=1}^m d_j^2$, with $c_0=0$. & Eq.~\eqref{eq:dimension-constants}\\[1pt]
$K_n,C_n$ & Centered constant $K_n=2c_{\lfloor n/2\rfloor}$ and
unrestricted constant $C_n=c_{\lfloor n/2\rfloor}+c_{\lceil n/2\rceil}$. & Sec.~\ref{sec:dimension}; App.~\ref{sec:unrestricted}\\[1pt]
$w,D,\Psi_n$ & In the centered equality construction,
$w_j=f_j-f_{j-1}$ for $1\leq j\leq n$, where
$f_j=\sqrt{\min(j,n-j)}$,
$D=\diag(w)$, and $\Psi_n(X)=\Pi(X)+\Tr(DX)D/n$. & Sec.~\ref{sec:projection-envelopes}\\[1pt]
$H,Y,\gamma$ & Supplied nonzero traceless Hermitian input, its image
$Y=S(H)$, and input quality $0<\gamma\leq1$ in
$\|S(H)\|_2\geq\gamma\|S\|_{2\to2}\|H\|_2$. & Alg.~\ref{alg:supplied-witness}\\[1pt]
$P_k,Q_l$,\newline $R_k,T_l,r_k,s_l$ & In Algorithm~\ref{alg:supplied-witness},
ordered spectral-prefix projections, their complements when ranks exceed
$n/2$, and ranks $r_k=\min(k,n-k)$, $s_l=\min(l,n-l)$. & Sec.~\ref{sec:intro-contributions}\\[1pt]
$d_{kl},d_*$ & Candidate score $|\Tr[T_lS(R_k)]|/\sqrt{r_ks_l}$
and its maximum returned by Algorithm~\ref{alg:supplied-witness};
$d_*=0$ when $S(H)=0$. & Alg.~\ref{alg:supplied-witness}\\[1pt]
$\Omega$ & Generic unital channel used to define a fixed-pair
score, including a supplied model channel. & Sec.~\ref{sec:witness-stability}\\[1pt]
$\eta,d_\Omega(P,Q)$ & Certified model error
$\|\Phi-\Psi\|_{2\to2}\leq\eta$ and the fixed-pair score
$|\Tr[Q(\Omega-\Pi)(P)]|/\sqrt{\rank(P)\rank(Q)}$. & Sec.~\ref{sec:witness-stability}\\[4pt]

\multicolumn{3}{@{}l}{\textit{Parallel uses and model channels}}\\[2pt]
$\Phi^{\otimes k},\Phi^k$ & $k$ independent parallel uses, acting on
$M_{n^k}(\mathbb C)$, and $k$ sequential applications, acting on
$M_n(\mathbb C)$. & Sec.~\ref{sec:parallel}\\[1pt]
$\delta_k,\epsilon_k$,\newline $\delta_\infty,\epsilon_\infty$ & Tensor-power discrepancies
$\delta(\Phi^{\otimes k})$, $\epsilon(\Phi^{\otimes k})$ and their
limits as $k\to\infty$ for a fixed channel. & Thm.~\ref{thm:regularized}\\[1pt]
$P_k,Q_k$,\newline $p_k,q_k,t_k$ & In the parallel threshold construction,
input/output threshold projections, marginal probabilities
$\tau_{n^k}(P_k),\tau_{n^k}(Q_k)$, and joint probability
$\tau_{n^k}[Q_k\Phi^{\otimes k}(P_k)]$. & Sec.~\ref{sec:parallel}\\[1pt]
$X,Y$ (scores) & In the regularization proof, real zero-mean,
unit-variance scores with $\mathbb E[XY]=\lambda$.
The numerical sampler uses the unscaled profile weights as scores. & Secs.~\ref{sec:parallel}, \ref{sec:balanced-effects}\\[1pt]
$A,\Phi_A$ & Doubly stochastic transition matrix with output index $i$
and input index $j$, and its measure-and-prepare channel. The joint
input/output law is $A_{ij}/n$. & Eq.~\eqref{eq:classical-channel}\\[1pt]
$\mathcal F_m,\mathcal L_n$ & Two-block averaging channel with block ranks
$m,n-m$, and the localized-mode channel with block ranks $1,n-1$ in
dimension $n$. & Prop.~\ref{prop:endpoint-envelopes}; Ex.~\ref{ex:cut-separation}\\[1pt]
$\mathcal B_\rho,\rho$ & Embedded binary symmetric channel with scalar
correlation $0\leq\rho\leq1$. In classical comparisons,
$\rho=\rho(A)$; in state formulas, $\rho$ is a density matrix. & Prop.~\ref{prop:binary-two-copy}\\[4pt]

\multicolumn{3}{@{}l}{\textit{Numerical comparisons}}\\[2pt]
$w_j,v,W_r,\beta$ & In the profile study, nonnegative weights $w_j$,
their signed extension $v$, partial sums $W_r=\sum_{j=1}^r w_j$,
and strength $0\leq\beta\leq1$. & Sec.~\ref{sec:spectral-profiles}\\[1pt]
$\Phi_{v,\beta},\delta_1$ & Profile channel
$\Phi_{v,\beta}(X)=\Pi(X)+\beta\Tr(DX)D/n$, with $D=\diag(v)$;
$\delta_1=\delta(\Phi)$ is the exact one-use optimum of the model under study.
& Secs.~\ref{sec:spectral-profiles}, \ref{sec:balanced-effects}\\[1pt]
$M,r,s,\sigma$ & In the qubit calibration, real Bloch matrix, input and
measurement Bloch vectors, and the three Pauli matrices. & Sec.~\ref{sec:qubit-calibration}\\[1pt]
$\kappa,g_\kappa,E_k,F_k$ & Tie-band tolerance $\kappa=10^{-12}$,
the odd threshold function, and balanced input/output effects on $k$ copies. & Sec.~\ref{sec:balanced-effects}\\[1pt]
$Z_k,\nu_k,\widehat\nu_k$ & Product of the thresholded score sums,
true balanced-effect value $\nu_k=|\mathbb EZ_k|/2$, and estimate
$\widehat\nu_k=|\overline Z_k|/2$. & Eq.~\eqref{eq:empirical-effect}\\[1pt]
$Z_k^{(\ell)},\overline Z_k$ & Threshold product on trajectory $\ell$ and
sample mean $\overline Z_k=N^{-1}\sum_{\ell=1}^N Z_k^{(\ell)}$. & Sec.~\ref{sec:balanced-effects}\\[1pt]
$N,M,r_N$ & In the sampling experiment, trajectories per model,
number of displayed settings, and simultaneous Hoeffding radius
$r_N=\sqrt{\log(2M/0.01)/(2N)}$. & Eq.~\eqref{eq:empirical-confidence}\\[1pt]
$T_X,T_Y$ & Counts of input/output score sums falling inside the tie band. & Table~\ref{tab:threshold-all-settings}\\[1pt]
$P_+,p_b,q_b$ & In the repeated-majority comparison, positive-profile
projection, single-use acceptance probability, and 64-use majority
acceptance probability, with $b=0$ for $\Pi$ and $b=1$ for $\Phi$.
& Eq.~\eqref{eq:repeated-bias}\\[4pt]

\multicolumn{3}{@{}l}{\textit{Auxiliary comparison notation}}\\[2pt]
$\operatorname{diam}(H)$,\newline $\operatorname{diam}(x)$ & Largest minus
smallest eigenvalue of a Hermitian matrix $H$, and largest minus
smallest coordinate of a real vector $x$. & Sec.~\ref{sec:dimension}\\[1pt]
$d,e(S,T)$ & Graph degree and number of ordered adjacent vertex pairs
in $S\times T$ in the classical graph comparison. & App.~\ref{sec:proof}\\[1pt]
$L_0(H),L(H)$ & Spectral-prefix costs for the centered and unrestricted
decompositions, bounded by $\sqrt{K_n}\|H\|_{\mathrm F}$ and
$\sqrt{C_n}\|H\|_{\mathrm F}$, respectively. & Eqs.~\eqref{eq:centered-cost}, \eqref{eq:spectral-cost}\\[1pt]
$h(S),h_{\mathrm{cl}}(B)$ & Heights
$h(S)=\sqrt{\|S\|_{1\to1}\|S\|_{\infty\to\infty}}/\|S\|_{2\to2}$
and its classical counterpart using induced vector norms; defined for
nonzero maps and matrices. & App.~\ref{sec:proof}\\[1pt]
$S^\sharp,S_{\mathrm R},S_{\mathrm I}$ &
$S^\sharp(X)=S(X^\dagger)^\dagger$;
$S_{\mathrm R}=(S+S^\sharp)/2$ and
$S_{\mathrm I}=(S-S^\sharp)/(2i)$ are Hermiticity-preserving components. & App.~\ref{sec:general}\\
\end{longtable}
\endgroup
\end{document}